\documentclass[sigconf, nonacm,screen]{acmart}

\newcommand\vldbdoi{XX.XX/XXX.XX}
\newcommand\vldbpages{XXX-XXX}
\newcommand\vldbvolume{14}
\newcommand\vldbissue{1}
\newcommand\vldbyear{2026}
\newcommand\vldbauthors{\authors}
\newcommand\vldbtitle{\shorttitle} 
\newcommand\vldbavailabilityurl{https://github.com/kingstenzzz/FlashOrder-paper}
\newcommand\vldbpagestyle{plain} 

\usepackage[T1]{fontenc}
\usepackage[ruled,vlined]{algorithm2e}
\DontPrintSemicolon
\SetAlgoNoEnd
\SetAlgoNoLine
\SetKwComment{Comment}{$\triangleright$ }{}
\SetAlgoCaptionSeparator{.}
\usepackage{graphicx}
\usepackage{textcomp}
\usepackage{xcolor}
\usepackage{booktabs}
\usepackage{tabularx}
\usepackage{array}
\usepackage{multirow}
\usepackage{amsmath}
\usepackage{enumitem}
\usepackage{amsthm}
\usepackage{stfloats}
\usepackage{threeparttable}

\theoremstyle{plain}
\newtheorem{proposition}{Proposition}[section]

\begin{document}

\title{Breaking Cycles for Scalable Fair Ordering in Blockchain Systems}

\author{Jinchun He}
\affiliation{%
  \institution{Beihang University}
}
\email{hejinchun@buaa.edu.cn}

\author{Wangjie Qiu}
\affiliation{%
  \institution{Beihang University}
}
\email{wangjieqiu@buaa.edu.cn}

\author{Yizhong Liu}
\affiliation{%
  \institution{Beihang University}
}
\email{liuyizhong@buaa.edu.cn}

\author{Shengda Zhuo}
\affiliation{%
  \institution{Jinan University}
}
\email{zhuosd96@gmail.com}

\author{Kwok-Yan Lam}
\affiliation{%
  \institution{Nanyang Technological University}
}
\email{kwokyan.lam@ntu.edu.sg}

\begin{abstract}
In blockchain systems, transaction order directly determines financial outcomes: unfair ordering enables front-running and sandwich attacks that have extracted over \$686M from Ethereum users. Current fair-ordering protocols aggregate pairwise receive-order evidence from replicas. Under contention or adversarial manipulation, however, Condorcet cycles force them into global strongly connected component (SCC) condensation, causing delays, coarse batches, and scaling failures.

We present FlashOrder, a deterministic fair-ordering engine that localizes cyclic ambiguity before it propagates across the batch. FlashOrder embeds pairwise preferences into one-dimensional canonical positions, clusters nearby transactions with a partition hypergraph, and performs hierarchical inter- and intra-cluster serialization, replacing batch-wide SCC condensation with localized sorting and aggregation. 
Evaluated against Themis (CCS '23) and Rashnu (VLDB '24) on a libhotstuff-based prototype, FlashOrder achieves up to 10.5$\times$ higher throughput than Themis and 4.8$\times$ higher than Rashnu, with the latency gap widening as network scales. In controlled adversarial simulation, it reduces maximum rank displacement by 88.7\%, and under Condorcet attacks it sustains 12.0$\times$ and 9.7$\times$ higher throughput than Themis and Rashnu on average. These results show that localizing cyclic ambiguity yields stronger fairness at substantially higher throughput.
\end{abstract}

\keywords{blockchain, Byzantine consensus, transaction sequencing, order fairness}
 
\maketitle
\pagestyle{\vldbpagestyle}
\begingroup\small\noindent\raggedright\textbf{PVLDB Reference Format:}\\
\vldbauthors. \vldbtitle. PVLDB, \vldbvolume(\vldbissue): \vldbpages, \vldbyear.\\
\href{https://doi.org/\vldbdoi}{doi:\vldbdoi}
\endgroup
\begingroup
\renewcommand\thefootnote{}\footnote{\noindent
This work is licensed under the Creative Commons BY-NC-ND 4.0 International License. Visit \url{https://creativecommons.org/licenses/by-nc-nd/4.0/} to view a copy of this license. For any use beyond those covered by this license, obtain permission by emailing \href{mailto:info@vldb.org}{info@vldb.org}. Copyright is held by the owner/author(s). Publication rights licensed to the VLDB Endowment. \\
\raggedright Proceedings of the VLDB Endowment, Vol. \vldbvolume, No. \vldbissue\ %
ISSN 2150-8097. \\
\href{https://doi.org/\vldbdoi}{doi:\vldbdoi} \\
}\addtocounter{footnote}{-1}\endgroup

\ifdefempty{\vldbavailabilityurl}{}{
\vspace{.3cm}
\begingroup\small\noindent\raggedright\textbf{PVLDB Artifact Availability:}\\
The source code, data, and/or other artifacts have been made available at \url{\vldbavailabilityurl}.
\endgroup
}

\section{Introduction}
\label{sec:intro}

Replicated transactional systems rely on consensus protocols to ensure that all replicas execute transactions in the same deterministic order~\cite{thomson2010determinism}. Blockchain systems and distributed databases share this design space as replicated transactional systems~\cite{ruan2021blockchain}, where Byzantine Fault-Tolerant (BFT) protocols guarantee safety and liveness~\cite{castro1999practical,yin2019hotstuff} despite adversarial participants. However, while existing BFT protocols~\cite{buchman2018tendermint,miller2016honeybadger} ensure ordering consistency, they fall short in guaranteeing \textit{fairness}: a malicious leader can still manipulate transaction sequencing to censor, reorder, or prioritize transactions without violating consensus correctness~\cite{daian2020flash,kelkar2020order,vafadar2023condorcet}. 
This problem is particularly critical in decentralized finance (DeFi), where transaction order directly determines economic outcomes. Adversarial ordering manipulation, commonly known as maximal extractable value (MEV), has extracted more than \$686M on Ethereum alone~\cite{daian2020flash,qin2021flashloans,li2023demystifying}. More broadly, whenever multiple transactions compete for limited shared resources, ordering itself determines fairness. Consequently, recent systems have introduced \emph{order fairness} guarantees~\cite{kelkar2020order,vafadar2023condorcet}: if sufficiently many honest replicas observe transaction \(t_i\) before \(t_j\), the committed order should preserve this precedence relation.

\noindent\textbf{Existing order fairness systems rely on global dependency-graph resolution.} Several technical approaches have been proposed to prevent ordering manipulation. Timestamp-based schemes such as Pompe \cite{zhang2020pompe} and Wendy \cite{kursawe2020wendy} assign each transaction a representative timestamp and sort by temporal ordering, but remain vulnerable to clock manipulation and do not tolerate Byzantine timestamp injection. Trusted-execution-environment (TEE) designs \cite{stathakopoulou2021tee} enforce ordering through hardware isolation, trading stronger guarantees for reliance on secure hardware assumptions. Cryptographic approaches such as proof-carrying fair ordering \cite{ren2025proofcarrying} provide verifiable fairness guarantees at the cost of additional computational overhead. A more widely adopted family of solutions---and the focus of this paper---builds on pairwise receive-order evidence: each replica records its local transaction receive order, and the protocol aggregates these pairwise precedence relations into a dependency graph to derive a globally fair sequence. This evidence model is natural for Byzantine settings because it does not require synchronized clocks, trusted hardware, or strong cryptographic assumptions.

\begin{table*}[t]
    \centering
    \begin{threeparttable}
    \caption{Comparison with representative fair-ordering protocols.}
    \label{tab:intro-comparison}
    \begin{tabular}{@{}llllll@{}}
        \toprule
        \textbf{Protocol} & \textbf{Fairness} & \textbf{Graph Scope} & \textbf{Cycle Resol.} & \textbf{SCC-Dep.} & \textbf{Condorcet Handling} \\
        \midrule
        Wendy~\cite{kursawe2020wendy} & timed-relative-fairness & --- & Timestamp comparison & No & Timing-sensitive \\
        Aequitas~\cite{kelkar2020order} & $\gamma$-batch-order-fairness & Full batch & Global SCC condensation & Yes & Global collapse \\
        Themis~\cite{kelkar2023themis} & $\gamma$-batch-order-fairness & Full batch & Global SCC + batch unspooling & Yes & Global collapse \\
        Rashnu~\cite{nagda2024rashnu} & $\gamma$-batch-order-fairness & Pruned (data-dep.) & Pruned SCC condensation & Yes & Pruned collapse \\
        \midrule
        \textbf{FlashOrder} & \textbf{$\gamma$-batch-order-fairness} & \textbf{Partition clusters} & \textbf{Local clustering + sorting} & \textbf{No} & \textbf{Localized degradation} \\
        \bottomrule
    \end{tabular}
    \begin{tablenotes}
        \small
        \item \emph{Graph Scope}: the scope of the dependency structure used for ordering.
        \item \emph{Cycle Resol.}: how cyclic pairwise ambiguity is resolved.
        \item \emph{SCC-Dep.}: whether the protocol depends on SCC condensation.
        \item \emph{Condorcet Handling}: qualitative behavior under Condorcet-style cyclic attacks.
    \end{tablenotes}
    \end{threeparttable}
\end{table*}

Within this pairwise-evidence family, a sequence of systems has progressively refined how the dependency graph is constructed and integrated with consensus. Themis \cite{kelkar2023themis} first implemented practical order-fairness: each replica records its local transaction receive order, the protocol aggregates these pairwise precedence relations into a directed dependency graph, resolves cyclic ambiguity through strongly connected component (SCC) condensation, and produces a linear sequence via topological sorting. Rashnu \cite{nagda2024rashnu} observes that not all transaction pairs require ordering constraints---only data-dependent transactions accessing shared state need fair ordering---and therefore constructs dependency edges only for conflicting transactions to reduce graph volume. 

\noindent\textbf{Global SCC resolution limits fair-ordering scalability.}
Despite these refinements in graph construction, all of these systems share the same SCC-based cycle-resolution paradigm---first construct a global dependency graph, then resolve cycles through SCC condensation---and are therefore subject to the same bottleneck. Under high contention, network jitter, or adversarial Condorcet attacks, pairwise majority relations become densely cyclic. Once such cyclic ambiguity propagates across the batch, SCC condensation moves onto the critical path, producing coarse-grained batching, unstable latency, and poor scalability.
This bottleneck is acute in the operating regimes that matter most for modern blockchain systems. High-frequency DeFi workloads generate many near-simultaneous transactions whose pairwise precedence is weakly separated; adversaries can further amplify this ambiguity by selectively inducing conflicting local observations. In these settings, SCC-based resolution causes confirmation delay (cyclic transactions wait for the full SCC), coarse batches (a small ambiguous region collapses a much larger batch), and poor scalability (global cycle resolution becomes the critical path as replicas and transactions grow). Recent bounded-unfairness results sharpen this concern by showing that minimizing order distortion inside such cyclic regions is NP-hard and admits no constant-factor polynomial-time approximation \cite{kiayias2024ordering}. The central weakness of SCC-based fair ordering is thus structural, not merely an implementation artifact.

This observation motivates a different design question: rather than accelerating global cycle resolution, can we structure the ordering pipeline so that cyclic ambiguity never reaches the scale at which such resolution becomes expensive? 

\noindent\textbf{Our approach.} We present \textbf{FlashOrder}, a deterministic fair-ordering engine that answers this question with a \emph{localization-based} design. The engine operates through a three-stage pipeline: (1)~robust canonical embedding of thresholded pairwise evidence into a one-dimensional ordering signal, (2)~adaptive clustering of nearby ambiguous transactions as a partition hypergraph, and (3)~hierarchical inter- and intra-cluster serialization. FlashOrder does not redefine fairness and does not redesign the consensus substrate; instead, it rewrites the ordering engine itself so that dense pairwise evidence is processed through localized sorting and aggregation rather than batch-wide SCC condensation.

\textbf{Contributions:} This paper makes the following contributions.
\begin{itemize}[leftmargin=*]
    \item \textbf{We identify global SCC-based cycle resolution as the main execution bottleneck in prior graph-based fair-ordering protocols.} Our analysis shows that, under contention and Condorcet-style cyclic ambiguity, the practical cost of fair ordering is dominated not by pairwise evidence aggregation itself, but by the global procedure required to resolve cyclic precedence relations; recent bounded-unfairness results further indicate that optimizing order distortion inside such cyclic regions is structurally intractable in general \cite{kiayias2024ordering}.
    \item \textbf{We propose FlashOrder, a deterministic fair-ordering engine that localizes cyclic ambiguity before it spreads across the batch.} FlashOrder combines thresholded canonical embedding, adaptive clustering over a partition hypergraph, and hierarchical serialization to transform dense pairwise evidence into a localized ordering structure, replacing global SCC condensation with deterministic local sorting and aggregation.
    \item \textbf{We show that this localized execution path simultaneously improves fairness and performance across complementary evaluation settings.} We implement FlashOrder on top of libhotstuff and evaluate it against HotStuff \cite{yin2019hotstuff}, Themis \cite{kelkar2023themis}, and Rashnu \cite{nagda2024rashnu} under varying network scales, contention levels, fairness settings, and adversarial ordering-manipulation scenarios. At $n=9$, FlashOrder achieves $30{,}680$ txs/s, a $10.5\times$ improvement over Themis and $4.8\times$ over Rashnu. As the network scales to $n=101$, the gap further widens: FlashOrder maintains $180$\,ms latency while Themis reaches $779$\,ms and Rashnu reaches $401$\,ms, because FlashOrder keeps the ordering path localized whereas Themis and Rashnu incur increasingly expensive global graph processing. In controlled adversarial simulation at $\gamma=0.8$, FlashOrder reduces maximum rank displacement from 301 (Themis) and 287 (Rashnu) to 34 positions. Under 50\,ms WAN Condorcet attacks at $\gamma=0.6$, FlashOrder sustains $12.0\times$ higher throughput than Themis and $9.7\times$ higher than Rashnu on average. Together, these results suggest that FlashOrder lowers the practical cost of fair ordering while improving fairness quality under explicit cluster-order conditions.
\end{itemize}

\section{Background}
\label{sec:background}

\subsection{Fair Ordering Preliminaries}
\label{sec:bg_prelim}

Replicated transactional systems rely on consensus protocols to ensure that all honest replicas execute transactions in the same deterministic order~\cite{thomson2010determinism}. Byzantine Fault-Tolerant (BFT) protocols guarantee safety and liveness~\cite{castro1999practical,yin2019hotstuff} but do not inherently ensure \textit{order fairness}: a malicious leader can still manipulate transaction sequencing without violating consensus correctness~\cite{daian2020flash,kelkar2020order}. This has motivated \emph{order fairness} as an additional property: if sufficiently many honest replicas observe transaction $tx_i$ before $tx_j$, the committed order should preserve this precedence~\cite{kelkar2020order,vafadar2023condorcet}.

\noindent\textbf{Why receive-order evidence?}
Several approaches have been proposed to realize order fairness. Timestamp-based schemes such as Pompe~\cite{zhang2020pompe} and Wendy~\cite{kursawe2020wendy} assign each transaction a representative timestamp, but remain vulnerable to clock manipulation and Byzantine timestamp injection. Propagation-time estimation cannot be captured precisely because the network is asynchronous and transactions may be arbitrarily delayed. As a result, the most widely adopted family of fair-ordering protocols---including Aequitas~\cite{kelkar2020order}, Themis~\cite{kelkar2023themis}, and Rashnu~\cite{nagda2024rashnu}---builds on \emph{pairwise receive-order evidence}: each replica records the order in which it receives transactions, and the protocol aggregates these local observation sequences to derive a globally fair ordering. This model is natural for Byzantine settings because it requires neither synchronized clocks nor trusted hardware.

\noindent\textbf{Condorcet paradox.}
Even when every individual replica's local observation sequence is transitive, the aggregated pairwise majority preferences can form cycles, a phenomenon known as the Condorcet paradox from social choice theory~\cite{condorcet1785}. Consider three replicas observing three transactions:
\begin{itemize}[leftmargin=*,nosep]
    \item Replica 1: $tx_1 \prec tx_2 \prec tx_3$
    \item Replica 2: $tx_2 \prec tx_3 \prec tx_1$
    \item Replica 3: $tx_3 \prec tx_1 \prec tx_2$
\end{itemize}
A majority (2 of 3) observe $tx_1$ before $tx_2$, a majority observe $tx_2$ before $tx_3$, and a majority observe $tx_3$ before $tx_1$---forming an intransitive cycle $tx_1 \to tx_2 \to tx_3 \to tx_1$. No linear order can respect all three majority preferences simultaneously, making strict pairwise fairness impossible in general. To address this impossibility, $\gamma$-batch-order-fairness~\cite{kelkar2020order} relaxes the requirement: transactions involved in a Condorcet cycle are delivered in the same batch, with no ordering guarantee within the batch. A deterministic tie-breaking rule (e.g., lexicographic ordering by transaction identifier) is then applied within each batch to produce a total order. Section~\ref{sec:problem} provides the formal definition.

\subsection{From Fairness to Bottleneck}
\label{sec:bg_bottleneck}

State-of-the-art fair ordering protocols (e.g., Themis~\cite{kelkar2023themis}) enforce $\gamma$-batch-order-fairness by constructing a pairwise dependency graph $G$, where directed edges represent honest-majority preferences, and then executing Tarjan's algorithm to collapse strongly connected components (SCCs) into indivisible batches. While theoretically sound, this \emph{graph-centric topology} exhibits significant practical fragility under contention.

Under high-contention workloads or targeted Condorcet attacks~\cite{vafadar2023condorcet}, adversarial nodes deliberately broadcast conflicting transaction sequences. As illustrated in Fig.~\ref{fig:attack}, the attack proceeds in four stages: (1) geographically distributed clients inject transactions with conflicting preference sequences; (2) these conflicting sequences induce local order disagreements across replicas; (3) the aggregated pairwise preferences form an intransitive voting cycle (Condorcet paradox); and (4) the SCC-based protocol must collapse all transactions into a single indivisible batch, leading to severe latency degradation or unresponsive ordering.

Recent work has attempted to address this bottleneck at the communication layer (CFTO~\cite{nassar2024cfto}) or the definition layer (Age-Aware Fairness~\cite{sokolik2025age}). However, these refinements remain dependent on constructing dense pairwise dependency graphs and executing expensive SCC condensation. Excessively large batches in these systems can worsen frontrunning risks within the batch~\cite{park2025frontrunning}, weakening the fairness objective they intend to improve.

The root cause of this fragility is that SCC-based methods allow cyclic dependencies to propagate globally before attempting resolution. In contrast, FlashOrder localizes cyclic conflicts into small, bounded regions before they chain across the entire transaction set, reducing the cost of cycle resolution from traversing a dense $O(m^2)$ graph to sorting within and across a small number of compact clusters.

\begin{figure}[t!]
    \centering
    \includegraphics[width=\linewidth]{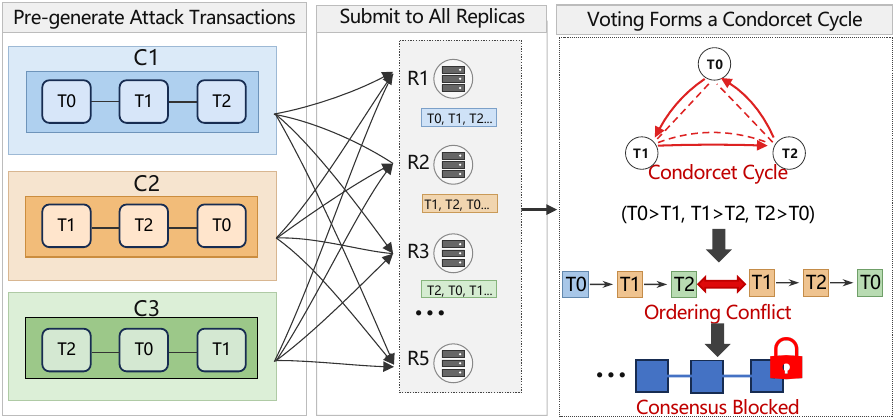}
    \caption{Condorcet attack: conflicting client sequences induce cyclic preferences, forcing SCC-based protocols to collapse the entire batch.}
    \label{fig:attack}
\end{figure}

\section{System Model and Problem Setup}
\label{sec:system_model}

This section provides the system model of FlashOrder. We first specify the BFT system and adversarial model, then state the threshold pairwise fairness objective adopted from prior work.

\subsection{System Model and Threat Model}
\label{sec:threat_model}

We consider a standard partially synchronous BFT network model inherited from classical Byzantine Fault Tolerant consensus protocols \cite{castro1999practical,yin2019hotstuff,miller2016honeybadger,gilad2017algorand,team2018avalanche,buchman2018tendermint,keidar2021dagrider}. The system consists of $n$ known replicas, of which at most $f$ may exhibit arbitrary (Byzantine) behavior at any given time. Byzantine replicas may deviate from the protocol arbitrarily, including sending conflicting messages to different subsets of honest replicas, omitting messages, or colluding to manipulate the transaction ordering.

The network follows the partial synchrony model: message delivery delays are bounded by some unknown $\Delta$ after an unknown Global Stabilization Time (GST), but there is no bound before GST. This is consistent with the underlying consensus protocol (HotStuff) and with prior fair-ordering baselines (Themis, Rashnu). Each replica maintains a local observation sequence $\mathcal{O}_i$ representing the order in which it receives transactions from clients. The observation sequences may differ due to network latency, message reordering, or adversarial manipulation by Byzantine nodes.

We assume standard cryptographic primitives: digital signatures for message authentication, collision-resistant hash functions for message commitment, and a public-key infrastructure (PKI) for key distribution. The adversary is computationally bounded and cannot break these cryptographic assumptions or compromise honest replica identities.

\noindent\textbf{Threat Model.} Following prior work on order-fairness \cite{vafadar2023condorcet,kelkar2023themis}, we consider Byzantine replicas that coordinate to maximize ordering disruption by sending conflicting transaction sequences to different honest replicas. This can induce Condorcet-style cyclic pairwise majorities in the aggregated preference graph~\cite{condorcet1785}. Such cyclic regions are the critical stress case for SCC-based fair-ordering engines because they prevent direct linearization of robust pairwise preferences and force cycle resolution onto the critical path.

Traditional BFT consensus requires $n > 3f$ to guarantee safety and liveness \cite{castro1999practical,yin2019hotstuff}. However, order-fairness introduces additional constraints because the protocol must distinguish honest-majority orderings from adversarial noise. Following Themis \cite{kelkar2023themis} and Aequitas \cite{kelkar2020order}, FlashOrder uses a stricter threshold.

\begin{lemma}
\label{lem:bft_threshold}
\textbf{BFT Feasibility Threshold for Order-Fairness.} Given $n$ replicas with at most $f$ Byzantine, the fairness threshold requires the following sufficient feasibility condition:
\begin{equation}
n > \frac{4f}{2\gamma - 1},
\end{equation}
where $\gamma \in (\frac{1}{2} + \frac{2f}{n}, 1]$ is the order-fairness parameter.
\begin{proof}
With $n$ replicas and $f$ Byzantine, the protocol relies on a quorum of $n-f$ replicas. Among this quorum, at most $f$ may be malicious and $f$ may be honest but slow. Thus, the fairness argument must reserve enough margin above a simple majority to absorb both Byzantine evidence and delayed honest evidence. Writing this feasibility slack as $(2\gamma-1)n$, order-fairness requires the slack to dominate the $4f$ uncertainty budget, i.e., $(2\gamma-1)n > 4f$. Rearranging yields $n > 4f/(2\gamma - 1)$. Under this convention, smaller $\gamma$ leaves less slack above the feasibility boundary, and therefore corresponds to a stricter admissible operating regime.
\end{proof}
\end{lemma}

\begin{lemma}
\label{lem:gamma_bounds}
\textbf{$\gamma$ Parameter Bounds.} The order-fairness parameter $\gamma$ is constrained to satisfy:
\begin{equation}
\frac{1}{2} + \frac{2f}{n} < \gamma \le 1.
\end{equation}
\end{lemma}
\begin{proof}
The lower bound ensures that the honest supermajority exceeds the adversarial majority: $\gamma > 1/2 + 2f/n$ derives from $n > 4f/(2\gamma-1)$. Under the threshold convention used in this paper, $\gamma$ should be read as a feasibility/slack parameter rather than as a direct agreement fraction: values closer to the lower bound impose a stricter evidence-admission threshold and require larger deployments, while $\gamma=1$ is the most permissive admissible setting.
\end{proof}

The tighter threshold $n > \frac{4f}{2\gamma-1}$ replaces the standard $n > 3f$ because order-fairness requires distinguishing the honest-majority signal from Byzantine noise. For typical configurations (e.g., $\gamma=0.8$, $f=15$), this requires $n > 100$, hence our experimental focus on $n=101$.

\subsection{Thresholded Pairwise Fairness Objective}
\label{sec:problem}

Let $\mathcal{T} = \{tx_1, tx_2, \ldots, tx_m\}$ denote a batch of $m$ transactions submitted to the BFT system. Each replica $k \in \{1, \ldots, n\}$ maintains an observation sequence $\mathcal{O}_k$ defining a total order over $\mathcal{T}$. We denote by $tx_i \prec_{\mathcal{O}_k} tx_j$ the event that transaction $tx_i$ precedes $tx_j$ in replica $k$'s observation sequence.

For any pair $(tx_i, tx_j) \in \mathcal{T} \times \mathcal{T}$, define the \emph{preference count}:
\begin{equation}
P_{ij} = |\{k \in [n] : tx_i \prec_{\mathcal{O}_k} tx_j\}|,
\end{equation}
which counts the number of replicas observing $tx_i$ before $tx_j$. The preference matrix $P = [P_{ij}]_{i,j \in [m]}$ captures all pairwise receive-order information.

\begin{definition}
\textbf{($\gamma$-Batch-Order-Fairness)~\cite{kelkar2020order,kelkar2023themis}.} A transaction ordering protocol is $\gamma$-batch-order-fair if, given a partition of $\mathcal{T}$ into ordered batches $\mathcal{B}_1, \mathcal{B}_2, \ldots, \mathcal{B}_k$, for any pair $tx_i, tx_j \in \mathcal{T}$:
\begin{equation}
\text{if } P_{ij} > \tau \text{ then } \beta(tx_i) \leq \beta(tx_j),
\end{equation}
where $\beta(tx)$ denotes the batch index containing $tx$, and the consensus threshold $\tau$ is:
\begin{equation}
\tau = \max(1, \Big\lfloor n(1-\gamma) \Big\rfloor + f + 1).
\end{equation}
\label{def:batch_order_fairness}
\end{definition}

Intuitively, $\gamma$-batch-order-fairness requires that robust pairs be correctly ordered across batches, but makes no pairwise guarantee within each batch. This definition, adopted from Aequitas~\cite{kelkar2020order} and used by Themis~\cite{kelkar2023themis}, reflects the practical limit of what pairwise-evidence protocols can guarantee: when Condorcet cycles exist (i.e., $P_{12} > \tau, P_{23} > \tau, P_{31} > \tau$ simultaneously), no linear order can satisfy all robust pairs, making strict pairwise $\gamma$-fairness unattainable in general. Under the threshold convention used here,
$\tau = \max(1, \lfloor n(1-\gamma) \rfloor + f + 1)$
is the evidence-admission threshold for declaring a pair robust. Thus smaller $\gamma$ values raise $\tau$ and admit only stronger pairwise receive-order evidence, making settings close to the feasibility lower bound the strictest operating regime. Conversely, $\gamma=1$ yields the lowest admission threshold ($\tau=f+1$) and is therefore the most permissive setting.

The key challenge, therefore, is not how to redefine fairness, but how to realize these thresholded pairwise constraints without placing global cycle condensation on the critical path.

\section{The FlashOrder Engine}
\label{sec:design}

FlashOrder is a cluster-based fair ordering engine designed as a plug-and-play, deterministic state machine for BFT consensus protocols. As illustrated in Fig.~\ref{fig:Framework}, the core pipeline comprises three phases: robust canonical embedding (RCE), adaptive hypergraph modeling (AHM), and hierarchical order resolution (HOR), which collectively replace global SCC condensation with localized sorting and aggregation.

\subsection{System Architecture and Processing Pipeline}
\label{sec:sys_arch}
Each replica maintains a local observation sequence and gossips it through the peer-to-peer layer. Once the leader collects a threshold quorum of these sequences, the consensus layer forwards the aggregated data to the fair ordering engine. The underlying BFT layer safely replicates the raw observation sequences, while FlashOrder operates as an independent deterministic state machine atop the agreed-upon data. Unlike prior decoupled systems (e.g., Themis) that still incur $O(m^2)$ global graph condensation, FlashOrder restructures the computational payload so that honest replicas derive the same global ordering under the same certified input.

The internal workflow of the \texttt{finalize()} engine processes transactions through three hierarchical phases, as formally defined in {Algorithm \ref{alg:flashorder}}. Each phase explicitly transforms transaction states to progressively resolve ordering dependencies. To make this replay property precise, FlashOrder applies fixed deterministic tie-breaking rules at every ordering stage: whenever canonical positions or intra-cluster scores are equal, transactions are ordered lexicographically by a unique transaction identifier, and whenever the quotient graph admits multiple valid topological orders, clusters are serialized by a deterministic cluster key.

\begin{figure*}[t]
    \centering
    \includegraphics[width=\linewidth]{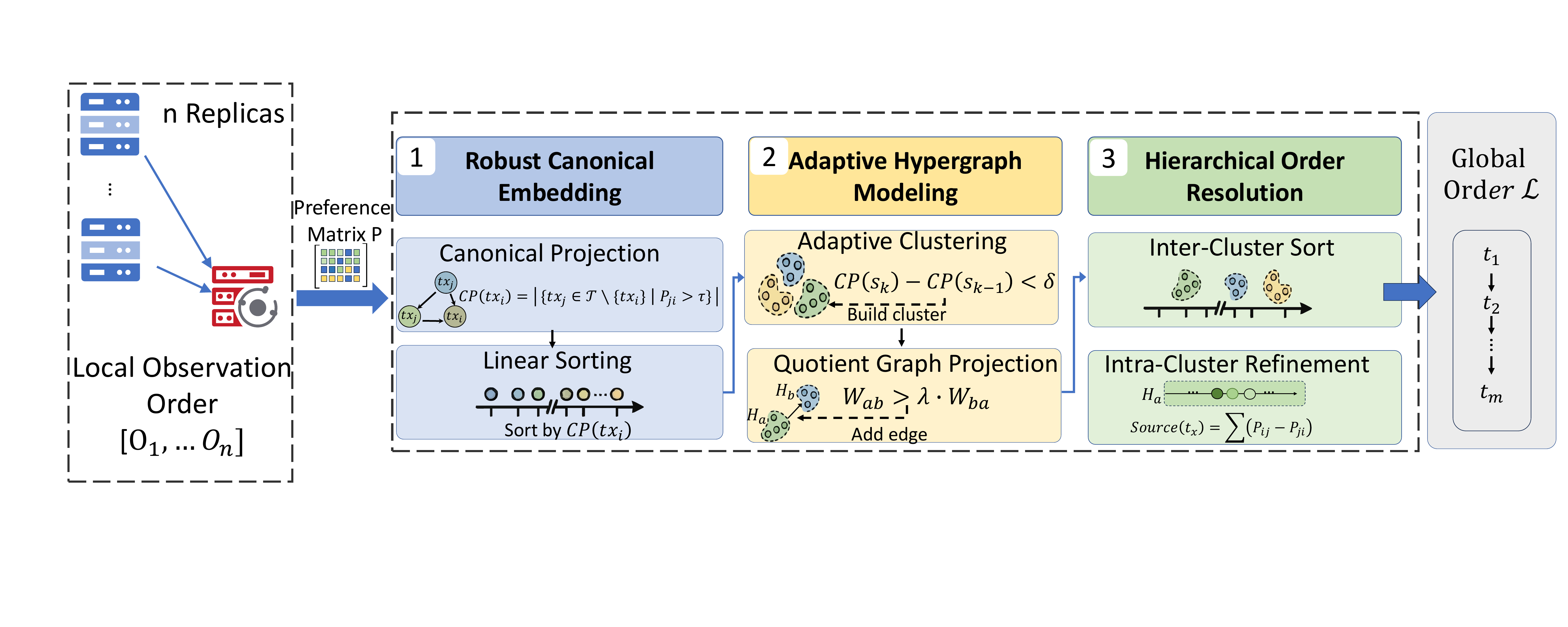}
    \caption{FlashOrder workflow}
    \label{fig:Framework}
\end{figure*}

\begin{algorithm}[!ht]
\caption{Main Workflow of FlashOrder}
\label{alg:flashorder}
\SetKwInOut{Input}{Input}\SetKwInOut{Output}{Output}
\Input{Transaction batch $\mathcal{T}$; agreed observations $\{\mathcal{O}_1, \dots, \mathcal{O}_n\}$; protocol params $\gamma, f, n$; system params $\Phi, \lambda, \mathbb{E}[delay], \Delta t_{gen}$}
\tcp{Each transaction carries a unique deterministic identifier $id(tx)$}
\Output{Global ordering $\mathcal{L}$}

\BlankLine
\tcp{Phase 1: Robust Canonical Embedding}
$\tau \leftarrow \max(1, \lfloor n \cdot (1 - \gamma) \rfloor + f + 1)$\;
$P_{ij} \leftarrow |\{v_k \mid tx_i \prec_{\mathcal{O}_k} tx_j\}|$\;
\For{each $tx_i \in \mathcal{T}$}{
  $CP(tx_i) \leftarrow |\{tx_j \in \mathcal{T} \setminus \{tx_i\} \mid P_{ji} > \tau\}|$\;
}
$\mathcal{S} = \langle s_1, \dots, s_m \rangle \leftarrow \text{sort } \mathcal{T} \text{ by } (CP \text{ asc.}, id(tx) \text{ asc.})$\;

\BlankLine
\tcp{Phase 2: Adaptive Hypergraph Modeling}
$\delta \leftarrow \Phi \cdot \frac{\mathbb{E}[delay]}{\Delta t_{gen}}$\;

\BlankLine
\tcp{Adaptive Geometric Clustering}
$\mathcal{H} \leftarrow \emptyset; H_{temp} \leftarrow \{s_1\}$\; 
\For{$k \leftarrow 2$ \KwTo $m$}{
    \If{$CP(s_k) - CP(s_{k-1}) < \delta$}{
        $H_{temp} \leftarrow H_{temp} \cup \{s_k\}$\;
    }
    \Else{
        $\mathcal{H} \leftarrow \mathcal{H} \cup \{H_{temp}\}; H_{temp} \leftarrow \{s_k\}$\;
    }
}
$\mathcal{H} \leftarrow \mathcal{H} \cup \{H_{temp}\}$\;

\BlankLine
\tcp{Quotient Graph Projection}
$\mathcal{G}_{HG} \leftarrow (\mathcal{H}, \mathcal{E}); \mathcal{E} \leftarrow \emptyset$\;
\For{each $(H_a, H_b) \in \mathcal{H} \times \mathcal{H}, a \neq b$}{
    $W_{ab} \leftarrow \sum_{tx_i \in H_a} \sum_{tx_j \in H_b} P_{ij}$\;
    \If{$W_{ab} > \lambda \cdot W_{ba}$}{ Add edge $H_a \rightarrow H_b$ to $\mathcal{E}$\; }
}
\While{nontrivial SCC $C$ exists in $\mathcal{G}_{HG}$, selected by deterministic cluster key}{
    $H_{merge} \leftarrow \bigcup_{H \in C} H$\;
    Replace vertices in $C$ with $H_{merge}$ in $\mathcal{G}_{HG}$\;
    Update edges $\mathcal{E}$ for $H_{merge}$\;
}
$\mathcal{H}_{final} \leftarrow V(\mathcal{G}_{HG})$\;
\For{each $H \in \mathcal{H}_{final}$}{
    $key(H) \leftarrow \min_{tx \in H} id(tx)$\;
}

\BlankLine
\tcp{Phase 3: Hierarchical Order Resolution}
$\mathcal{H}_{sorted} \leftarrow \text{LexicographicTopologicalSort}(\mathcal{G}_{HG}, key)$\;
$\mathcal{L} \leftarrow \langle \rangle$\;
\For{each $H \in \mathcal{H}_{sorted}$}{
    \For{each $tx_i \in H$}{
        $Score_i \leftarrow \sum_{tx_j \in H, j \neq i} (P_{ij} - P_{ji})$\;
    }
    $\mathcal{L}_H \leftarrow \text{sort } H \text{ by } (Score \text{ desc.}, CP \text{ asc.}, id(tx) \text{ asc.})$\;
    $\mathcal{L}.\text{append}(\mathcal{L}_H)$\;
}
\Return $\mathcal{L}$\;
\end{algorithm}

\subsection{Phase 1: Robust Canonical Embedding}
The primary objective of the first phase is to rapidly condense the raw, aggregated observation sequences into a robust 1D baseline array without allocating any expensive graph data structures (Algorithm \ref{alg:flashorder}, Phase 1). The input to this phase is the dense preference matrix $P$, where each entry $P_{ij}$ records the number of replicas that observed $tx_i$ arriving before $tx_j$. To filter out Byzantine manipulations and network jitter, the system uses a threshold $\tau = \max(1, \lfloor n(1-\gamma) \rfloor + f + 1)$, where $n$ is the total number of nodes and $f$ represents the upper bound of Byzantine faults. With this convention, decreasing $\gamma$ raises $\tau$ and therefore makes the robust-pair admission rule stricter.

Instead of inserting transactions into a directed graph to find edges, the engine calculates a \emph{Canonical Position (CP)} for every transaction via a sequential scan:
\begin{equation}
CP(tx_i) = |\{tx_j \in \mathcal{T} \setminus \{tx_i\} \mid P_{ji} > \tau\}|
\end{equation}
This Step 1: Canonical Projection requires only sequential scans over the preference matrix, yielding a deterministic rank proxy that counts how many other transactions robustly precede it. 

Once the CP array is populated, the deterministic state machine executes Step 2: Linear Sorting which sorts the transaction batch in ascending order of their CP values, using the transaction identifier as a canonical secondary key whenever CP values are equal, thereby yielding the linear baseline sequence $\mathcal{S}$. This canonical embedding compresses the $O(m^2)$ global dependency graph into a single linear array. By projecting the data in this manner, FlashOrder localizes any underlying cyclic conflicts (Condorcet paradoxes) to adjacent indices within this 1D array. This locality enables bounded conflict resolution in subsequent phases, without global graph condensation.

\subsection{Phase 2: Adaptive Hypergraph Modeling}
Transactions with proximate canonical coordinates in $\mathcal{S}$ represent ordering ambiguities that frequently manifest as Condorcet cycles under contention. Instead of instantiating a dense $O(m^2)$ dependency graph where every transaction is a node (the primary cause of computational saturation in prior works), FlashOrder applies a {Structural Compression} technique through adaptive geometric clustering (Algorithm \ref{alg:flashorder}, Phase 2), which produces a \emph{partition hypergraph} $(\mathcal{T}, \mathcal{H})$: the transaction set $\mathcal{T}$ is partitioned into disjoint clusters $\mathcal{H} = \{H_1, H_2, \ldots, H_k\}$, where each $H_a \subseteq \mathcal{T}$, $|H_a| \geq 1$, and $\bigcup_{a=1}^{k} H_a = \mathcal{T}$ with $H_a \cap H_b = \emptyset$ for $a \neq b$. We use the term partition hypergraph only for this disjoint cluster family; the subsequent inter-cluster dependency structure is an ordinary quotient graph. This reduces the topological state space from $m$ transactions to $k$ macro-nodes ($k \ll m$).

To achieve this, the engine first executes Step 1: Adaptive Clustering. Iterating through the sorted sequence $\mathcal{S}$, adjacent transactions are grouped into a cluster $H$ when their coordinate difference remains within threshold $\delta$. As defined in Algorithm~\ref{alg:flashorder}, we set
\begin{equation}
\delta = \Phi \cdot \frac{\mathbb{E}[delay]}{\Delta t_{gen}},
\end{equation}
where $\Phi$ is the fairness relaxation parameter (default $0.15$), and $\mathbb{E}[delay]$, $\Delta t_{gen}$ denote empirical network latency and transaction generation interval, respectively. This formulation keeps clustering sensitivity explicit and tunable through measurable system parameters. By confining cyclic dependencies within clusters, FlashOrder prevents these conflicts from growing into global SCCs.

Following consolidation, the system proceeds with Step 2: Quotient Graph Projection. A directed graph $\mathcal{G}_{HG} = (\mathcal{H}, \mathcal{E})$ is constructed as the quotient graph projection of the partition hypergraph, where nodes correspond to clusters $H_a \in \mathcal{H}$ and directed edges encode inter-cluster preference dominance. A directed edge $H_a \rightarrow H_b$ is established when aggregate preference weight $W_{ab} = \sum_{tx_i \in H_a} \sum_{tx_j \in H_b} P_{ij}$ exceeds reverse flow $W_{ba}$ by a specified margin. Formally, the margin $\lambda$ must satisfy $1 < \lambda < \frac{\gamma n - f}{(1-\gamma)n + f}$ so that inter-cluster edges are established only when the aggregate honest-majority signal reliably exceeds adversarial noise under the BFT threshold. In practice, $\lambda$ is a system parameter selected inside this feasible interval; lower-$\gamma$ regimes narrow the admissible interval and therefore require a smaller margin. In the rare case where macro-cycles persist in $\mathcal{G}_{HG}$ under extreme adversarial topologies, the engine detects nontrivial SCCs and deterministically merges each selected component into a single super-cluster, yielding an acyclic quotient graph. Each final cluster $H$ is then assigned a deterministic key $key(H)=\min_{tx \in H} id(tx)$.

\noindent\textbf{Derivation of the $\lambda$ bound.}
The $\lambda$ parameter controls when inter-cluster dominance is reliably detectable. We derive both bounds from $\gamma$-fairness requirements.

\noindent\emph{Lower bound ($\lambda > 1$):}
To establish a meaningful edge $H_a \rightarrow H_b$, we require clear dominance:
\begin{equation}
W_{ab} > \lambda \cdot W_{ba}
\end{equation}
If $\lambda \leq 1$, the condition $W_{ab} > \lambda \cdot W_{ba}$ would not guarantee $W_{ab} > W_{ba}$, allowing edges where reverse preference is equally strong. Thus, $\lambda > 1$ ensures inter-cluster edges require unambiguous forward dominance.

\noindent\emph{Upper bound derivation.}
Under the same feasibility parameter $\gamma$, the inter-cluster margin must account for the honest signal that can support the forward direction and the residual reverse evidence that may arise from disagreement, delay, or Byzantine behavior. Lower $\gamma$ values correspond to a stricter operating regime: the pairwise threshold $\tau$ is higher, and the admissible margin interval for $\lambda$ becomes narrower. The honest signal strength contributing to $W_{ab}$ is bounded by:
\begin{equation}
S_{\text{honest}} \approx \gamma (n - f) \approx \gamma n - f
\end{equation}
The adversarial noise or conflicting evidence that could contribute to $W_{ba}$ comes from: (i) the $(1-\gamma)$ fraction of honest nodes disagreeing, and (ii) all $f$ Byzantine nodes potentially providing reverse signals:
\begin{equation}
N_{\text{adversarial}} \approx (1-\gamma)(n - f) + f \approx (1-\gamma)n + f
\end{equation}
For the edge condition $W_{ab} > \lambda \cdot W_{ba}$ to reliably detect honest consensus over adversarial noise, we require:
\begin{align}
\lambda \cdot N_{\text{adversarial}} &< S_{\text{honest}} \nonumber \\
\lambda \cdot ((1-\gamma)n + f) &< \gamma n - f \nonumber \\
\lambda &< \frac{\gamma n - f}{(1-\gamma)n + f}
\end{align}

\noindent\emph{Feasibility condition:}
For this bound to exceed 1 (i.e., $\frac{\gamma n - f}{(1-\gamma)n + f} > 1$), we require:
\begin{align}
\gamma n - f &> (1-\gamma)n + f \nonumber \\
\gamma n - f - n + \gamma n &> f \nonumber \\
(2\gamma - 1)n &> 2f \nonumber \\
n &> \frac{2f}{2\gamma - 1}
\end{align}
Under the standard BFT threshold condition $n > \frac{4f}{2\gamma-1}$, this is always satisfied, guaranteeing the upper bound exceeds the lower bound. In the main deployment configurations we use $\lambda = 1.2$ as a conservative engineering margin; in stricter low-$\gamma$ regimes, the admissible interval narrows and $\lambda$ must be selected within the corresponding bound.

\subsection{Phase 3: Hierarchical Order Resolution}
The final phase of the FlashOrder pipeline deterministically flattens the dimensionally reduced data structures into a total ordering $\mathcal{L}$ through a two-tier serialization process (Algorithm \ref{alg:flashorder}, Phase 3). The prior phase absorbs dense, transaction-level Condorcet cycles into a compact macro-hypergraph ($\mathcal{G}_{HG}$), so this final serialization avoids the deep recursion stacks and overhead of global graph traversals.

First, the engine performs Step 1: Inter-Cluster Topological Sort. Because the prior phase has already absorbed any remaining macro-cycles through SCC merging, the quotient graph $\mathcal{G}_{HG}$ is guaranteed to be acyclic at this point. FlashOrder executes a lexicographic topological traversal over the macro-nodes (clusters), establishing a coarse-grained precedence among the aggregated transaction groups ($\mathcal{H}_{sorted}$). Among all currently eligible zero in-degree clusters, the engine always selects the one with the smallest deterministic key $key(H)=\min_{tx \in H} id(tx)$.

Next, the engine executes Step 2: Intra-Cluster Refinement to resolve the exact transaction sequence within each individual cluster. For a given cluster $H$, FlashOrder computes a \emph{Net-Preference Score} for every internal transaction $tx_i$:
\begin{equation}
Score(tx_i) = \sum_{tx_j \in H, j \neq i} (P_{ij} - P_{ji})
\end{equation}
This score sums the net preference for each transaction within the cluster. Because $|H|$ is small (bounded and typically orders of magnitude below the batch size $m$), this local scoring adds negligible cost. The transactions within the cluster are sorted in descending order of their net scores; if multiple transactions obtain the same score, ties are resolved first by ascending canonical position and then by ascending transaction identifier. Finally, the deterministic state machine concatenates these locally sorted segments to produce the global ordering $\mathcal{L}$. These deterministic tie-breaking rules are invoked only when the ordering evidence is exactly tied, and therefore do not reverse any inter-cluster precedence established by the quotient graph. Within each cluster, the net-preference score provides a deterministic refinement for transactions whose robust precedence relations are mutually contradictory or too weakly separated; its fairness quality is evaluated empirically rather than asserted as an unconditional guarantee inside arbitrary cyclic regions.

\subsection{Complexity Analysis}
At the level of pairwise evidence aggregation, any receive-order fairness protocol operating on a pairwise preference matrix must process $\Theta(m^2)$ relations in the worst case, where $m$ is the batch size. This is because there are $m(m-1)/2$ transaction pairs, and the protocol must collect and process evidence for each pair to determine the relative ordering under the fairness constraint. FlashOrder therefore remains quadratic in the worst case with respect to the input information, while reducing the practical cost of cycle resolution through localized serialization.

However, the empirical speedup over SCC-based protocols (e.g., Themis) stems from restructuring the \emph{systems-level overhead} and data access patterns. Traditional protocols resolve cycles by constructing dense, pointer-based directed graphs and executing Tarjan's SCC condensation. Under adversarial loads, these graphs become highly cyclic, causing high computational cost, unpredictable branch penalties, and deep recursion stacks.

In contrast, FlashOrder isolates the $\Omega(m^2)$ pairwise relation processing purely to sequential, efficient array arithmetic (Phase 1 and 2). It reduces the practical cost of the core topological cycle-resolution engine by replacing expensive, graph traversal with localized canonical sorting and clustering. By replacing combinatorial graph topology with linear 1D sorting, FlashOrder avoids the large constant factors that dominate traditional $O(m^2)$ fair ordering.

\section{Correctness Arguments}
\label{sec:correctness}

This section presents the correctness arguments for FlashOrder. We establish three formal properties: deterministic replayability as an unconditional guarantee, conditional cluster-order fairness under explicit separation conditions, and acyclic exactness in well-separated regimes.

A pair $(tx_i, tx_j)$ is considered \emph{robust} if the observed precedence count satisfies $P_{ij} > \tau$, where $\tau = \max(1, \lfloor n(1-\gamma) \rfloor + f + 1)$. Let $\mathcal{H}^*$ denote the final cluster family after quotient-graph SCC merging, and let $\beta(tx)$ be the position of the final cluster containing $tx$ in the lexicographic topological order.

\begin{proposition}
\label{thm:deterministic_replay}
\textbf{(Deterministic Replayability).}
For fixed parameters and the same certified observation input $\{\mathcal{O}_1,\ldots,\mathcal{O}_n\}$, all honest replicas running Algorithm~\ref{alg:flashorder} output the same total order $\mathcal{L}$.
\end{proposition}

\begin{proof}
The preference matrix $P$ is a pure function of the certified observation sequences. Canonical positions are computed from $P$ using a fixed threshold and are sorted with a deterministic identifier tie-breaker. Geometric clustering scans the sorted sequence deterministically. Quotient edges are computed from aggregate weights and a fixed margin, and nontrivial SCCs are merged using a deterministic cluster key. The final topological traversal and intra-cluster score sorting also use fixed tie-breaking rules. Hence every step is a deterministic function of the same input and parameters, so all honest replicas reproduce the same $\mathcal{L}$.
\end{proof}

\begin{proposition}
\label{thm:conditional_cluster_fairness}
\textbf{(Conditional Cluster-Order Fairness).}
Assume $n > \frac{4f}{2\gamma - 1}$. For every robust pair $(tx_i,tx_j)$ with $P_{ij}>\tau$, suppose the final clusters $H_i^*,H_j^*\in\mathcal{H}^*$ satisfy one of the following conditions: (i) $H_i^*=H_j^*$, or (ii) the final quotient graph contains a directed path from $H_i^*$ to $H_j^*$. Then FlashOrder satisfies the batch-level relation $\beta(tx_i)\leq\beta(tx_j)$ for that pair. If the condition holds for every robust pair, FlashOrder satisfies the corresponding cluster-level batch-order property.
\end{proposition}

\begin{proof}
If $H_i^*=H_j^*$, both transactions are in the same final cluster, so $\beta(tx_i)=\beta(tx_j)$. Otherwise, by assumption the final quotient graph contains a path from $H_i^*$ to $H_j^*$. The SCC-merge step leaves the quotient graph acyclic, and any topological ordering of an acyclic graph places the source of a directed path before its sink. Therefore $\beta(tx_i)<\beta(tx_j)$. Applying the same argument to every robust pair gives the cluster-level batch-order property.
\end{proof}

\begin{proposition}
\label{thm:acyclic_exactness}
\textbf{(Acyclic Exactness under Separation).}
If (a)~the robust preference graph is acyclic, (b)~the CP ordering is a valid topological order of that robust graph, (c)~the clustering threshold $\delta$ does not merge any pair separated by a robust edge, and (d)~the initial quotient graph (before SCC merging) is acyclic, then FlashOrder outputs a linear order satisfying every robust pair.
\end{proposition}

\begin{proof}
Condition~(c) ensures every robust edge connects distinct clusters. Condition~(d) ensures no SCC merging occurs, so the initial cluster assignment is the final one. By~(a) and~(b), the robust edges form a partial order consistent with the CP ordering; combined with~(c), each robust edge maps to a directed inter-cluster edge in the quotient graph whose direction is consistent with the robust preference. Since the quotient graph is acyclic~(d), the lexicographic topological sort returns a linear extension that preserves all quotient-graph edges and therefore all robust edges. Intra-cluster sorting cannot reverse any robust edge because, by~(c), no robust pair shares a cluster. Hence every robust pair appears in the output order in the required direction.
\end{proof}

FlashOrder also supports replayable order verification~\cite{mu2024speedyfair}: the leader attaches the $n-f$ signed observation sequences as a certificate alongside $\mathcal{L}$, and each honest replica replays the same certified input through its local engine. Any mismatch leads to rejection by the consensus layer.

\section{Evaluation}
\label{sec:evaluation}

\subsection{Experimental Setup}
\label{sec:eval_setup}
We implement FlashOrder on top of the open-source libhotstuff codebase~\cite{yin2019hotstuff} and compare it with HotStuff (no fairness), Themis~\cite{kelkar2023themis}, and Rashnu~\cite{nagda2024rashnu}. We select Themis and Rashnu as baselines because they are the only open-source, HotStuff-integrated fair-ordering engines sharing the same pairwise-evidence model, enabling apples-to-apples comparison.

Deployment experiments run on CloudLab~\cite{cloudlab} c6525-25g bare-metal machines (16-core AMD EPYC 7302P, 128\,GB RAM, 25\,Gbps Ethernet). We use the SmallBank benchmark with extreme Zipfian skew (SB\_SKEW=0.99) and high interaction probability (SB\_PROB=0.95), producing dense write-write conflicts that mirror real-world MEV arbitrage topologies. For fairness-quality evaluation (Section~\ref{sec:eval_fairness_quality}), we use algorithm-level simulation to evaluate larger network sizes and controlled adversarial sweeps beyond CloudLab scale.

\subsection{Evaluation Questions}
\begin{itemize}[leftmargin=*]
    \item \textbf{\S~\ref{sec:eval_deployment}}: Can FlashOrder deliver high throughput and low latency in an end-to-end BFT prototype?
    \item \textbf{\S~\ref{sec:eval_fairness_param}}: Does FlashOrder remain stable under stricter robust-pair admission?
    \item \textbf{\S~\ref{sec:eval_fairness_quality}}: Does FlashOrder preserve strong fairness quality under the same thresholded pairwise objective?
    \item \textbf{\S~\ref{sec:eval_condorcet}}: Does FlashOrder remain robust under adversarial ordering-manipulation attacks?
\end{itemize}

\subsection{End-to-End Performance in Deployment}
\label{sec:eval_deployment}
\subsubsection{Scaling with Replica Count}
\label{sec:eval_network}
\begin{figure}[!ht]
    \centering
    \includegraphics[width=\linewidth]{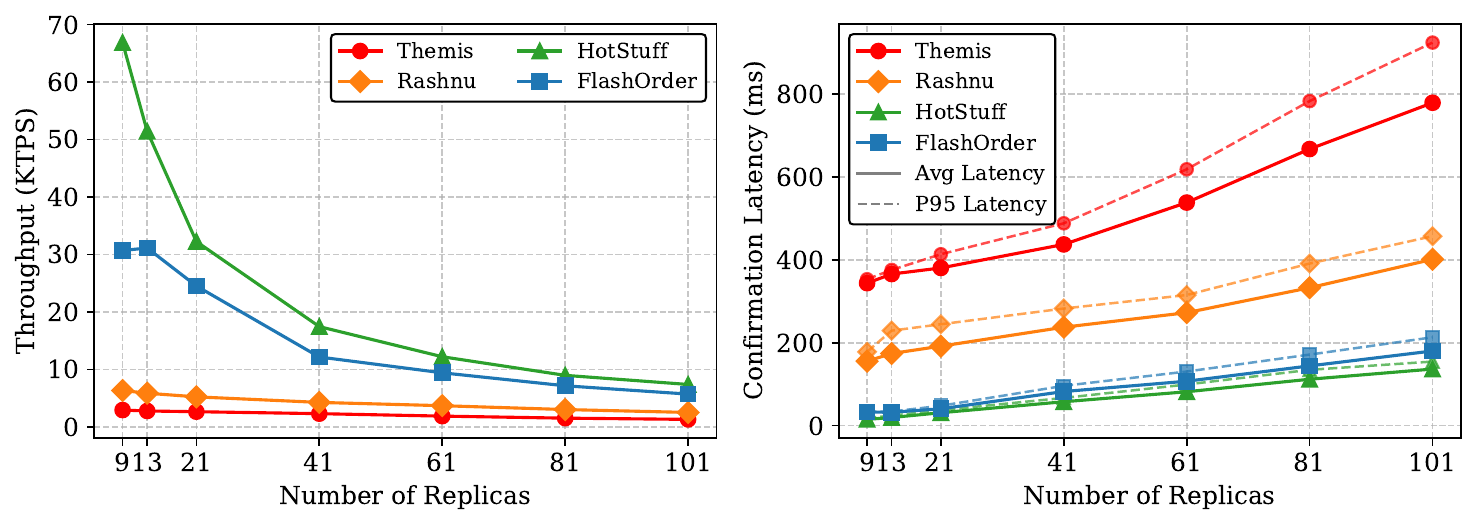}
    \caption{End-to-end throughput and latency under varying network scale ($n=9$ to $101$).}
    \label{fig:combined_performance}
\end{figure}

We evaluate whether FlashOrder preserves end-to-end practicality as the deployment scale grows. We vary the network size from $n=9$ to $n=101$ replicas and compare FlashOrder with the HotStuff baseline as well as the two state-of-the-art fair-ordering protocols, Themis and Rashnu. The results in Fig.~\ref{fig:combined_performance} show that FlashOrder maintains stable performance as the network grows, with a much smaller throughput gap to the non-fairness baseline than either Themis or Rashnu.

\textbf{Scalability analysis.} As shown in Fig.~\ref{fig:combined_performance}, the vanilla HotStuff baseline establishes the upper bound for performance because it does not enforce fair ordering. At $n=9$, HotStuff achieves a peak throughput of $66{,}853$ txs/s with an average latency of $14.9$ ms (p95: $18.1$ ms). FlashOrder follows closely, achieving $30{,}680$ txs/s with an average latency of $32.6$ ms (p95: $33.8$ ms). In this small-scale setting, FlashOrder achieves a $10.5\times$ throughput speedup over Themis ($2{,}920$ txs/s) and a $4.8\times$ improvement over Rashnu ($6{,}330$ txs/s). As the network expands to $101$ nodes, overall throughput naturally declines due to increased communication overhead in the BFT consensus. Even so, FlashOrder sustains $5{,}696$ txs/s compared to the $7{,}372$ txs/s of the baseline, translating to a throughput retention of $77.2\%$ at peak node count, while its average latency reaches $180.0$ ms (p95: $213.0$ ms, compared to the baseline's $136.4$ ms / $154.4$ ms). In contrast, Themis and Rashnu drop to $1{,}299$ txs/s and $2{,}504$ txs/s, respectively, with Themis's latency spiking to $779.0$ ms (p95: $924.2$ ms) and Rashnu reaching $401.1$ ms (p95: $457.0$ ms).

These results show that FlashOrder retains most of the practical efficiency of HotStuff while scaling much better than SCC-based fair-ordering engines. The gap widens with replica count: FlashOrder keeps the ordering path localized, while Themis and Rashnu pay for global graph processing as cyclic ambiguity grows.

\subsubsection{Performance under Challenging Operating Conditions}
\label{sec:eval_batch}

\begin{figure*}[!t]
    \centering
    \includegraphics[width=\linewidth]{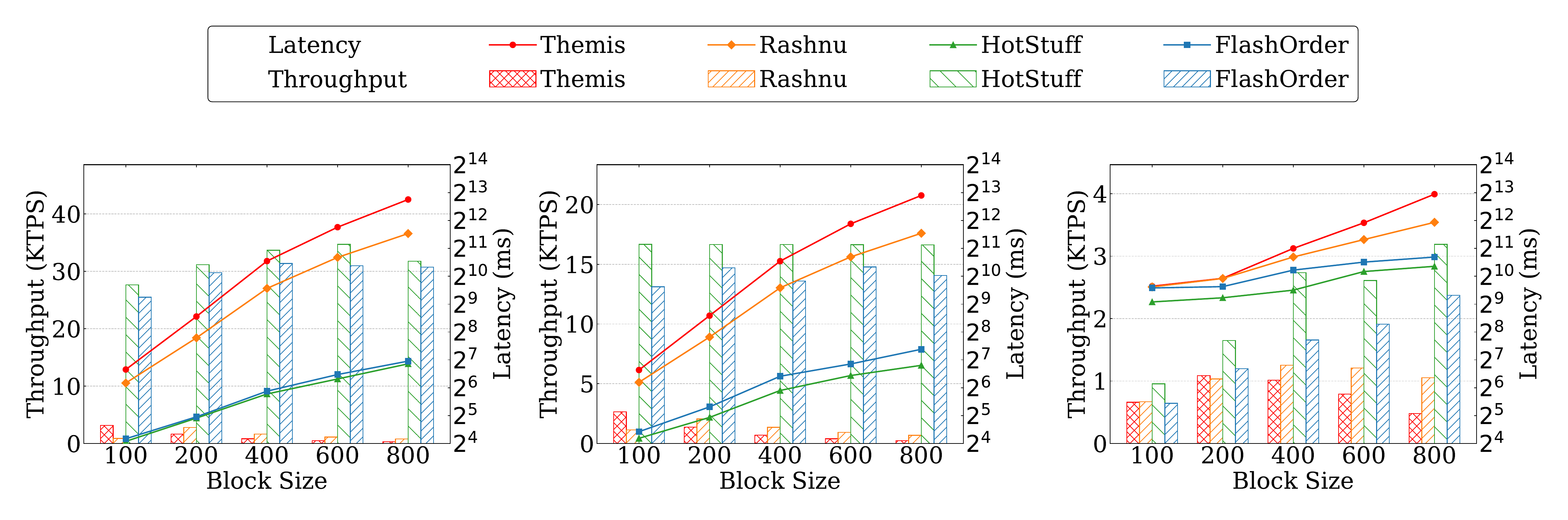}
    \caption{Throughput and latency under varying batch sizes ($n=21$): (a) Normal load, (b) Heavy contention load, (c) Geo-distributed setting with 50\,ms inter-node latency.}
    \label{fig:blocksize_performance}
\end{figure*}

We test whether FlashOrder's advantage persists precisely in the regimes that amplify cyclic ambiguity on the critical path. We evaluate the protocols under three scenarios with $n=21$ replicas, varying the block size from 100 to 800: (a)~normal load in LAN, (b)~heavy contention load in LAN ($SB\_SKEW{=}0.99, SB\_PROB{=}0.95$), and (c)~geo-distributed deployment with 50\,ms inter-node latency. As shown in Fig.~\ref{fig:blocksize_performance}, FlashOrder's advantage over Themis widens as contention and latency increase.
\textbf{Normal load.} Under normal load (Fig.~\ref{fig:blocksize_performance}a), FlashOrder closely tracks the HotStuff baseline, achieving $31.37$ KTPS at block size 400 ($93\%$ of baseline) with latency of $58.7$\,ms, compared to the baseline's $54.8$\,ms. Themis already exhibits significant degradation: at block size 800, Themis drops to $0.31$ KTPS with latency spiking to $6{,}905$\,ms.
\textbf{Heavy contention.} Under heavy contention (Fig.~\ref{fig:blocksize_performance}b), Themis's throughput degrades sharply: at block size 800, its throughput drops to $0.24$ KTPS with latency reaching $7,636$\,ms. FlashOrder maintains $13.60$ KTPS at block size 400 and achieves $166.6$\,ms latency at block size 800, delivering over $56\times$ higher throughput than Themis in this extreme stress test.
\textbf{Geo-distributed WAN.} In a geo-distributed setting with 50\,ms latency (Fig.~\ref{fig:blocksize_performance}c), the updated results show a clear separation at larger block sizes. At block size 800, FlashOrder reaches $2{,}373$ TPS with average latency $1{,}650$\,ms (p95: $3{,}608$\,ms), while Themis drops to $480$ TPS with $7{,}880$\,ms latency (p95: $11{,}667$\,ms) and Rashnu reaches $1{,}053$ TPS with $3{,}915$\,ms latency (p95: $4{,}883$\,ms). In this geo-distributed scenario, FlashOrder achieves $4.94\times$ and $2.25\times$ throughput gains over Themis and Rashnu, respectively. Compared with the HotStuff baseline at block size 800 ($3{,}189$ TPS, $1{,}311$\,ms), FlashOrder retains $74.4\%$ throughput with only $1.26\times$ latency overhead.

Taken together, these deployment results support the central systems claim of the paper: FlashOrder's advantage is largest when contention, batch size, and WAN delay make global SCC-style cycle resolution most expensive. Instead of allowing ambiguity to expand into a batch-wide graph bottleneck, FlashOrder confines it to localized clustering and serialization.

\subsection{Sensitivity to Order-Fairness Parameter}
\label{sec:eval_fairness_param}

\begin{figure*}[!t]
    \begin{minipage}{0.32\linewidth}
        \centering
        \includegraphics[width=\linewidth]{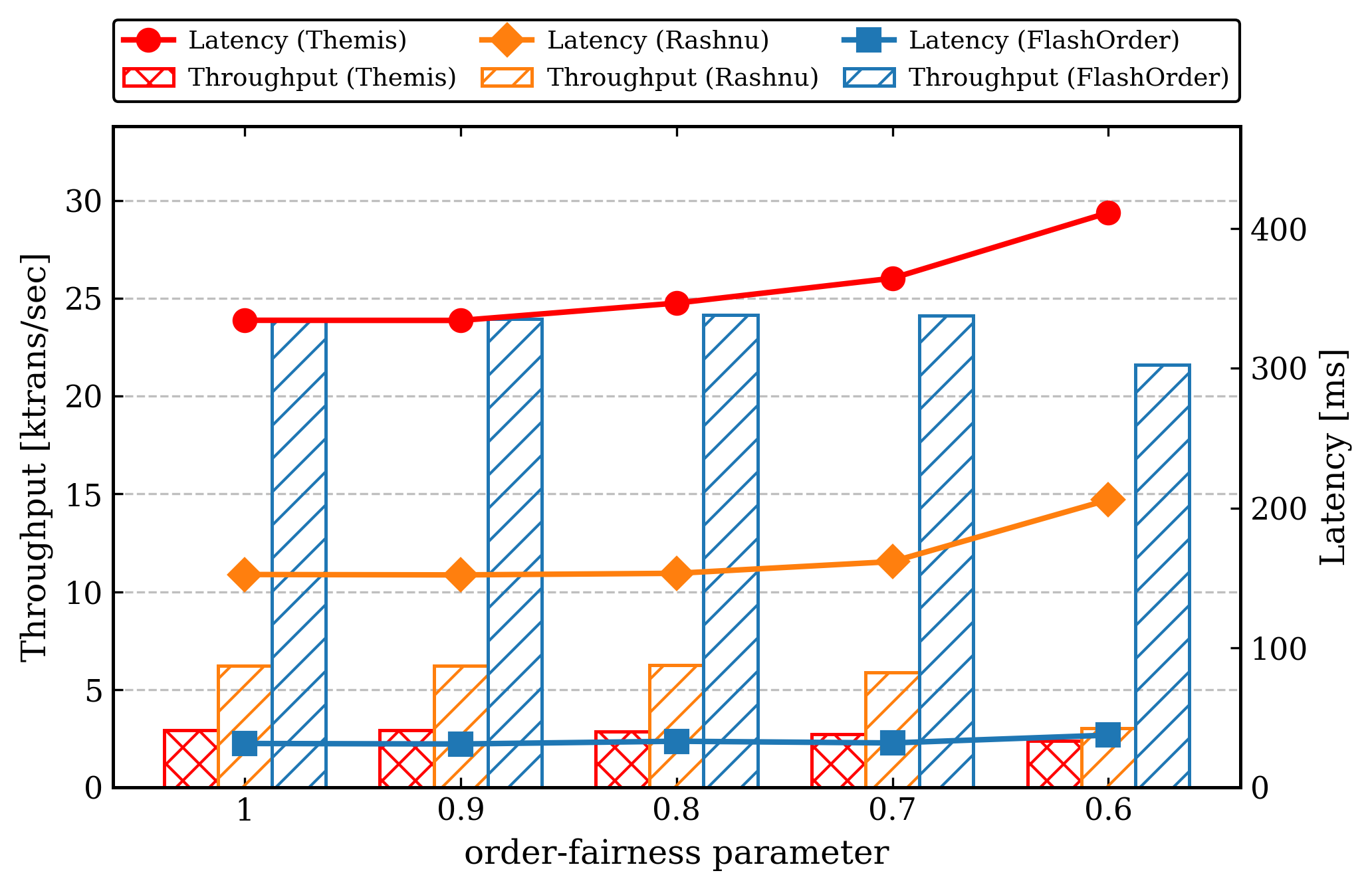}
        \caption{Throughput and latency under varying order-fairness parameter $\gamma$.}
        \label{fig:plot_fairness_combo}
    \end{minipage}\hfill
    \begin{minipage}{0.32\linewidth}
        \centering
        \includegraphics[width=\linewidth]{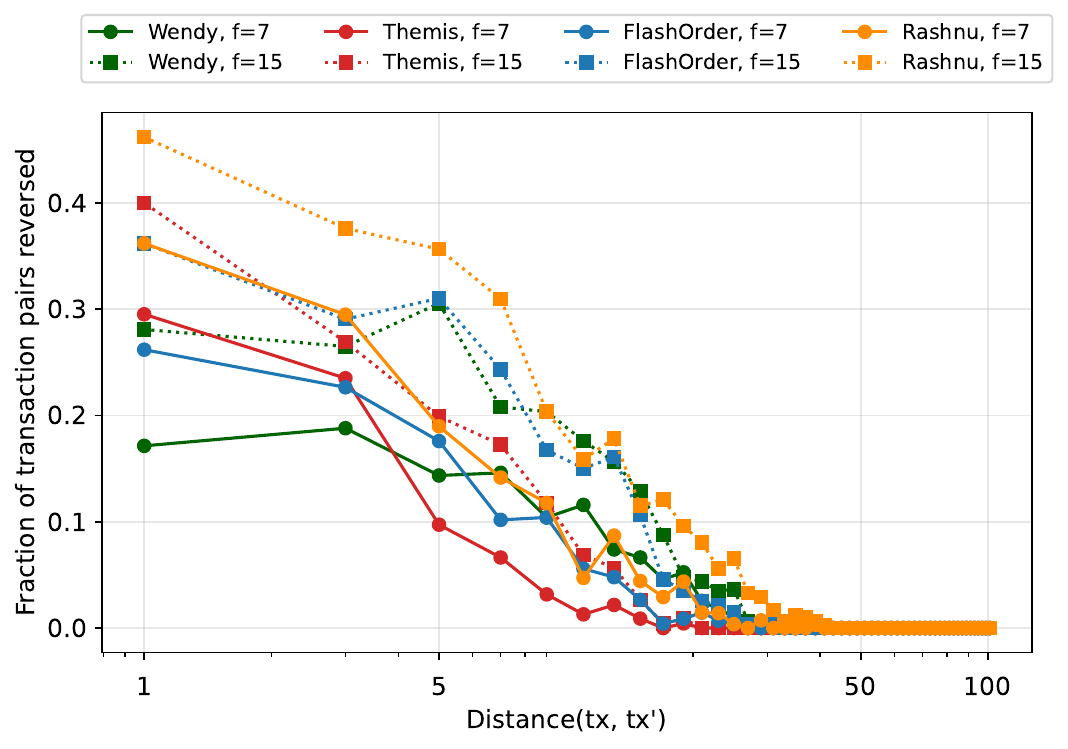}
        \caption{Ordering resilience and defensive stability under direct manipulation.}
        \label{fig:adv_reorder}
    \end{minipage}\hfill
    \begin{minipage}{0.32\linewidth}
        \centering
        \includegraphics[width=\linewidth]{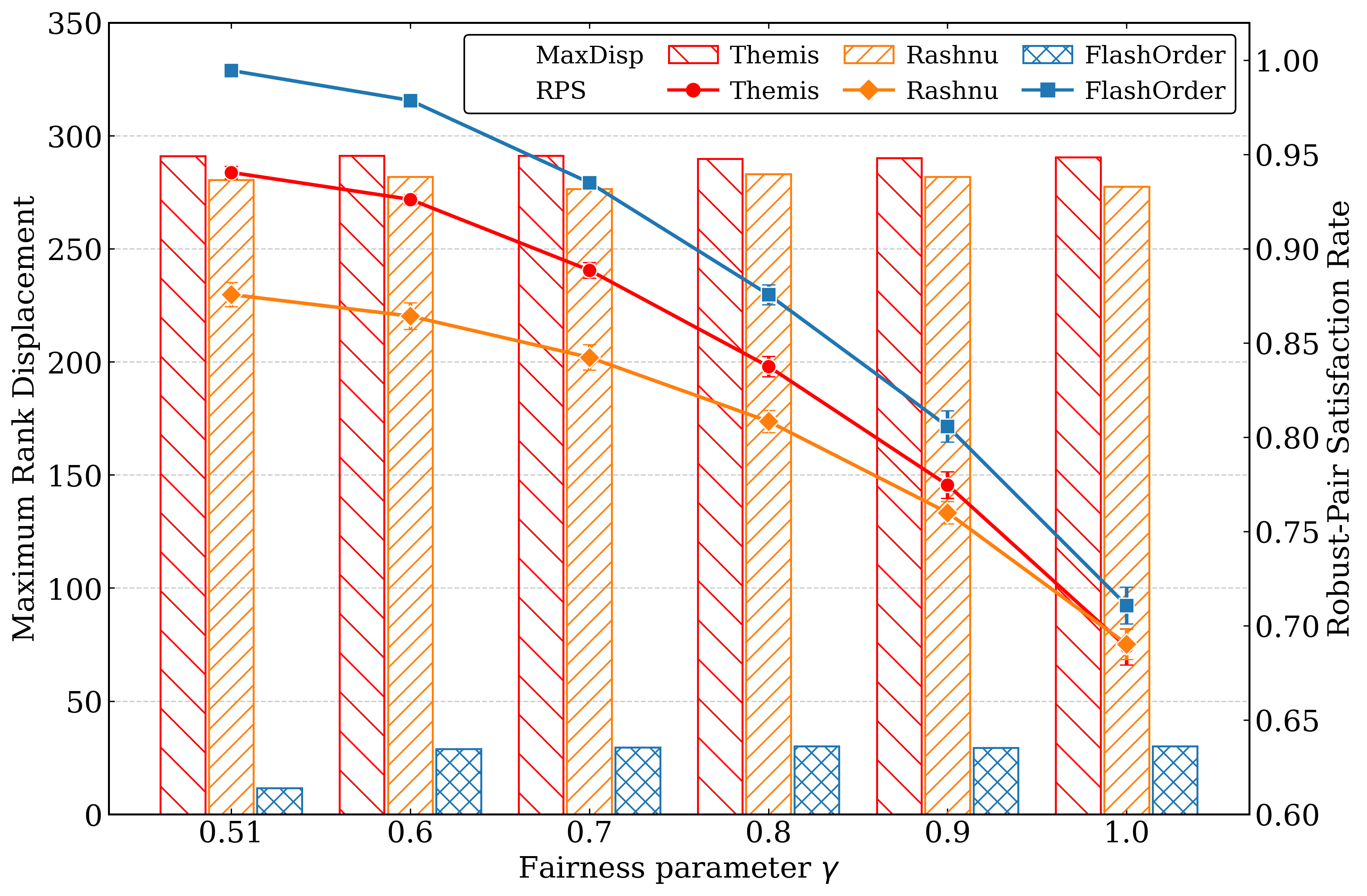}
        \caption{Adversarial fairness metrics ($f{=}2$): RPS (right axis) and MaxDisp (left axis).}
        \label{fig:fairness_metrics}
    \end{minipage}
\end{figure*}

We evaluate whether FlashOrder remains stable as lower $\gamma$ values impose stricter robust-pair admission in the thresholded pairwise precedence structure. We test FlashOrder, Themis, and Rashnu under identical network settings while progressively scaling $\gamma$ down from 1.0 to 0.6. Under the threshold $\tau=\lfloor n(1-\gamma)\rfloor+f+1$, lowering $\gamma$ raises the evidence threshold required to declare a precedence relation robust. This stricter regime reduces tolerance for weak receive-order evidence and makes the ordering engine more sensitive to the remaining high-confidence cyclic structure produced by contention and Condorcet-style manipulation. As depicted in Fig.~\ref{fig:plot_fairness_combo}, the performance of Themis degrades sharply as $\gamma$ decreases: its throughput drops from $3.0\times 10^3$ txs/s to $2.4\times 10^3$ txs/s, while its latency spikes non-linearly from 334.3\,ms to over 411.4\,ms. Although Rashnu reduces the sorting volume by constructing graphs only for data-dependent transactions, it still relies on graph traversal and cycle elimination when the stricter threshold leaves hard high-confidence cycles in the dependency structure. Consequently, Rashnu incurs approximately a $35\%$ latency penalty (reaching 205.8\,ms) under severe topological stress.

In contrast, FlashOrder shows strong performance stability. Even under the stricter robust-pair admission regime induced by $\gamma=0.6$, FlashOrder maintains a burst throughput of over $21.6\times 10^3$ txs/s, with its end-to-end latency remaining within 31.4--37.3\,ms. This result indicates that the localized execution path remains stable even when the protocol admits only stronger pairwise evidence: rather than paying for larger global cycle-resolution routines, FlashOrder continues to resolve the remaining ambiguity within bounded clusters.
\subsection{Fairness Quality under Thresholded Pairwise Constraints}
\label{sec:eval_fairness_quality}

We evaluate whether FlashOrder preserves strong fairness quality under the same thresholded pairwise objective used by prior work. We report two complementary metrics: the robust-pair satisfaction rate (RPS), which measures how often robust precedence constraints are respected, and maximum rank displacement (MaxDisp), which captures worst-case output distortion. Together, these metrics let us evaluate both constraint satisfaction and the bounded-unfairness style quantity that captures how much positional distortion an adversary can force \cite{kiayias2024ordering}.

\subsubsection{Robust-Pair Satisfaction Rate}
Beyond throughput and latency, we directly quantify fairness quality using two complementary metrics. The first is the \emph{Robust-Pair Satisfaction Rate (RPS)}, which measures the fraction of strongly agreed-upon transaction pairs correctly ordered in the protocol output. Formally, a pair $(tx_i, tx_j)$ is called \emph{robust} if $P_{ij} > \tau$, where $\tau = \max(1, \lfloor n(1-\gamma) \rfloor + f + 1)$ is the consensus threshold. Given a protocol output $\sigma^*$, the robust-pair satisfaction rate is defined as:
\begin{align}
\text{RPS} &= \frac{|\{(i,j) : P_{ij} > \tau \;\land\; \sigma^{*}(i) < \sigma^{*}(j)\}|}{|\{(i,j) : P_{ij} > \tau\}|}, \label{eq:rps}\\
&\quad \text{with }\text{RPS} := 1 \text{ if denominator is } 0. \nonumber
\end{align}
We simulate an adversarial setting with $f=2$ Byzantine nodes and 30 inducing transactions, reporting results at $\gamma=0.8$ (recommended operating point) and $\gamma=0.51$ (strictest BFT-admissible setting, requiring $n > 400$).
\begin{table}[t]
\centering
\caption{Fairness metrics under adversarial simulation ($f{=}2$, 30 inducing transactions, 20 trials).}
\label{tab:fairness}
\vspace{2pt}
\begin{small}
\begin{tabular}{@{}ll rrrrrr@{}}
\toprule
Protocol & $\gamma$ & RPS$\uparrow$ & MaxDisp$\downarrow$ & Cov. & Cov.$_{c}$ & MaxGrp$\downarrow$ \\
\midrule
Themis & 0.80 & $0.813$ & $301$ & $0.212$ & $0.212$ & $339$ \\
Rashnu & 0.80 & $0.790$ & $287$ & $0.462$ & $0.816$ & $4$ \\
\textbf{FlashOrder} & \textbf{0.80} & $\mathbf{0.873}$ & $\mathbf{34}$ & $\mathbf{0.815}$ & $\mathbf{0.812}$ & $\mathbf{54}$ \\
\midrule
Themis & 0.51 & $0.915$ & $300$ & $0.894$ & $0.895$ & $59$ \\
Rashnu & 0.51 & $0.853$ & $291$ & $0.507$ & $0.915$ & $2$ \\
\textbf{FlashOrder} & \textbf{0.51} & $\mathbf{0.996}$ & $\mathbf{10}$ & $\mathbf{0.910}$ & $\mathbf{0.908}$ & $\mathbf{70}$ \\
\midrule
Themis & 1.00 & $0.672$ & $301$ & $0.185$ & $0.186$ & $339$ \\
Rashnu & 1.00 & $0.675$ & $288$ & $0.385$ & $0.676$ & $5$ \\
\textbf{FlashOrder} & \textbf{1.00} & $\mathbf{0.704}$ & $\mathbf{87}$ & $\mathbf{0.590}$ & $\mathbf{0.588}$ & $\mathbf{136}$ \\
\bottomrule
\multicolumn{7}{@{}l}{\footnotesize Cov.: quotient-/condensation-path coverage of all robust pairs.}\\
\multicolumn{7}{@{}l}{\footnotesize Cov.$_{c}$: coverage restricted to conflicting pairs. MaxGrp: largest cluster/SCC size.}\\
\end{tabular}
\end{small}
\end{table}

\textbf{RPS advantage.} As shown in Table~\ref{tab:fairness}, FlashOrder achieves the highest RPS across all configurations. At the recommended $\gamma=0.8$, FlashOrder's RPS of $0.873$ represents a $7.4\%$ relative improvement over Themis ($0.813$) and a $10.5\%$ improvement over Rashnu ($0.790$). At the strict $\gamma=0.51$ setting with $n=401$ nodes, FlashOrder attains near-perfect robust-pair satisfaction with RPS $= 0.996$, while Themis and Rashnu remain at $0.915$ and $0.853$.

\textbf{Maximum Rank Displacement.}
The second fairness metric is \emph{Maximum Rank Displacement (MaxDisp)}, the worst-case positional shift between the protocol output and ground-truth ordering. In our simulation, the ground-truth reference order is the ascending send-time order used to generate the attack scenario. Whereas RPS measures aggregate fairness compliance, MaxDisp captures the worst-case attack surface available to an adversary.

\textbf{MaxDisp advantage.} MaxDisp reveals the most striking advantage of FlashOrder. At $\gamma=0.8$, FlashOrder's worst-case displacement is $34$ positions ($8.5\%$ of the 400-transaction batch), compared to $301$ ($75.3\%$) for Themis and $287$ for Rashnu, an $88.7\%$ reduction (i.e.\ $(301{-}34)/301$) in the attacker's maximum exploitable displacement, or equivalently an $8.9\times$ reduction. At $\gamma=0.51$ with $n=401$, this advantage grows to $30\times$ ($10$ vs $300$ positions). This is the same qualitative notion of positional distortion targeted by bounded-unfairness formulations \cite{kiayias2024ordering}; FlashOrder improves it not by searching for an optimal order inside a global SCC, but by preventing ambiguity from expanding into such a region in the first place.

\textbf{Structural decomposition.}
To understand \emph{why} FlashOrder achieves these gains, Table~\ref{tab:fairness} decomposes each protocol's robust-pair handling into two structural mechanisms: (1)~\emph{absorbed}---robust pairs that land in the same cluster or SCC, where ordering information is effectively lost; and (2)~\emph{covered} (Cov.)---robust pairs whose clusters/SCCs differ but the quotient or condensation DAG preserves a one-way path, indicating genuine structural preservation. The MaxGrp column reports the size of the largest cluster (FlashOrder) or SCC (Themis/Rashnu). At $\gamma=0.8$, Themis collapses $339$ out of $400$ transactions ($85\%$) into a single SCC (MaxGrp$=339$), so its Cov. is only $0.212$---most robust pairs are absorbed rather than genuinely preserved. FlashOrder's largest cluster contains only $54$ transactions ($13.5\%$), yet achieves Cov.$=0.815$ through quotient-path coverage. Rashnu keeps SCCs tiny (max $4$) by restricting edge construction to conflicting pairs, but leaves $54\%$ of all robust pairs structurally uncovered (Cov.$=0.462$); its conflict-restricted coverage Cov.$_{c}=0.816$ slightly exceeds FlashOrder's $0.812$ ($0.4\%$ difference, within statistical noise), which is expected: Rashnu exclusively optimizes for conflicting pairs while FlashOrder's hypergraph construction must simultaneously resolve non-conflicting robust pairs. FlashOrder achieves this near-parity on conflict pairs \emph{while} extending structural coverage to all robust pairs (Cov.$=0.815$ vs.\ Rashnu's $0.462$). Under the strict $\gamma=0.51$ setting, Themis's SCC collapse becomes far less severe (MaxGrp drops to $59$) and its Cov. improves to $0.894$, confirming that stricter robust-pair admission naturally limits SCC size; however, FlashOrder still leads with Cov.$=0.910$ and near-perfect RPS$=0.996$. At $\gamma=1.0$, all protocols see RPS decline ($0.67$--$0.70$), but structural differences persist: Themis again collapses to MaxGrp$=339$ with Cov.$=0.185$, while FlashOrder maintains Cov.$=0.590$ ($3.2\times$ higher) and MaxDisp$=87$ ($3.5\times$ lower than Themis's $301$), confirming that the quotient-path advantage holds across all parameter regimes.

\textbf{Scalability effect.} Moving from the recommended setting $(\gamma=0.8)$ to the stricter BFT-admissible setting $(\gamma=0.51, n=401)$, FlashOrder's fairness lead further amplifies: RPS increases from $0.873$ to $0.996$, while MaxDisp decreases from $34$ to $10$. Conversely, Themis's MaxDisp remains at approximately $300$ regardless of network scale. This suggests that FlashOrder not only avoids the execution bottleneck of SCC-based ordering, but also converts stricter robust-pair admission and larger-scale evidence aggregation into better fairness quality. Fig.~\ref{fig:fairness_metrics} illustrates how these metrics vary continuously across different $\gamma$ values.

\subsection{Robustness under Condorcet Attacks}
\label{sec:eval_condorcet}

\textbf{Attack setting.} We further evaluate adversarial service robustness under explicit ordering-manipulation attacks in both simulation and real prototype deployment.

For the simulation-side attack study, we configure a 101-node network with $f=15$ Byzantine nodes and $\gamma=0.8$. \emph{This satisfies the BFT threshold $n > \frac{4f}{2\gamma-1} = \frac{60}{0.6} = 100$, placing FlashOrder within the threshold range considered by our fairness analysis.} The Byzantine nodes systematically broadcast conflicting inverted transaction sequences to distinct subsets of honest nodes, increasing the length and density of dependency cycles. For the deployment-side attack study, we use a geo-distributed setting with 50\,ms inter-node latency and inject explicit Condorcet cycles by splitting malicious clients into three groups that submit contradictory local orders (A$\!>\!$B, B$\!>\!$C, C$\!>\!$A) while background SmallBank traffic continues. We sweep attack cycle length $\{3,10,20\}$ and report two representative fairness endpoints, $\gamma\in\{0.6,1.0\}$.

\textbf{Ordering resilience.} According to the distance metric
\[
Dist(tx_i, tx_j)=\left|\#\{r: tx_i \prec_{\mathcal{O}_r} tx_j\}-\#\{r: tx_j \prec_{\mathcal{O}_r} tx_i\}\right|,
\]
which measures the receive-order vote margin between the two relative orders, the protocol's ability to preserve the relative ordering of honest transactions is highlighted by the sharp decline in reversal rates as distance increases. In the simulation, transaction pairs are bucketed by the equivalent minority-count representation and displayed using the corresponding vote-margin labels. As shown in Fig.~\ref{fig:adv_reorder}, at the most ``fragile'' point of $Dist=1$, FlashOrder limits successfully flipped pairs to 49 in standard scenarios and 63 under peak adversarial load. This protection is maintained while achieving a 23.15$\times$ speedup in execution time over Themis, reducing latency from 13.015\,s to 0.562\,s. By utilizing a structural cluster graph to avoid the cycle explosion associated with $O(m^2)$ binary dependency graphs, FlashOrder helps prevent computational saturation during active attacks.

\begin{figure}[!ht]
    \centering
    \includegraphics[width=\linewidth]{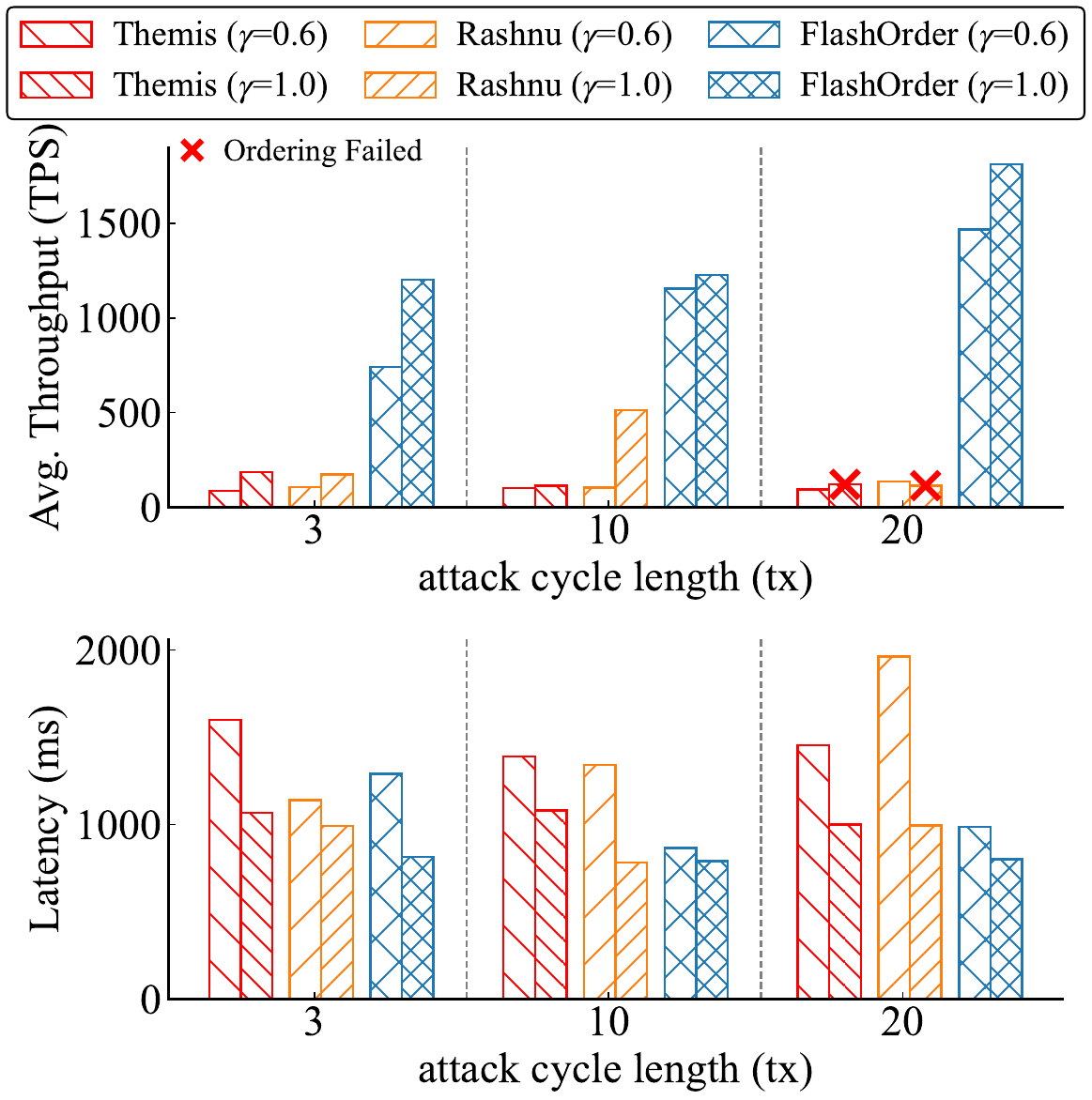}
    \caption{Condorcet attack robustness in WAN (50\,ms): average throughput (top) and client latency (bottom) versus attack cycle length for $\gamma\in\{0.6,1.0\}$.}
    \label{fig:condorcet_attack}
\end{figure}

\textbf{Throughput under attack.} Figure~\ref{fig:condorcet_attack} reports the 50\,ms WAN results for the two representative endpoints used in the paper, $\gamma\in\{0.6,1.0\}$, corresponding to a conservative threshold and the maximally permissive threshold, respectively. At $\gamma=0.6$, FlashOrder sustains 742/1155/1467 TPS for cycle lengths 3/10/20, whereas Themis reaches 87/100/93 TPS and Rashnu reaches 106/104/136 TPS. Averaged across cycle lengths, FlashOrder achieves 1121 TPS, i.e., $12.0\times$ higher than Themis and $9.7\times$ higher than Rashnu under adversarial Condorcet attacks. At $\gamma=1.0$, FlashOrder reaches 1202/1228/1813 TPS versus Themis 186/113/120 and Rashnu 172/513/113, corresponding to average gains of $10.1\times$ over Themis and $5.3\times$ over Rashnu.

\textbf{Latency and failure modes.} Latency trends are consistent with the throughput advantage. FlashOrder keeps client latency in a lower band than Themis and Rashnu in almost all settings: 867--1292\,ms at $\gamma=0.6$ and 790--814\,ms at $\gamma=1.0$, compared with Themis at 1389--1600\,ms ($\gamma=0.6$) / 1002--1082\,ms ($\gamma=1.0$) and Rashnu at 1139--1963\,ms ($\gamma=0.6$) / 784--995\,ms ($\gamma=1.0$). The only notable outlier is FlashOrder at $(\gamma=0.6,\,\text{cycle}=3)$, where latency rises to 1292\,ms, yet throughput still remains far above both baselines. We also observe explicit ordering failures in the same 50\,ms sweep: at $(\gamma=1.0,\,\text{cycle}=20)$, Themis abandons fair ordering once the Condorcet attack breaks its tournament-graph assumption, while Rashnu triggers fallback after over-pruning the graph rather than completing its intended ordering procedure.

Taken together, the attack experiments show that FlashOrder degrades more gracefully than SCC-based baselines along both performance and fairness dimensions. Under direct ordering manipulation, FlashOrder both preserves more robust precedence relations and bounds worst-case displacement more tightly, while still sustaining much higher service capacity than Themis and Rashnu.

\section{Related Work}
\label{sec:related}

The proliferation of MEV in DeFi \cite{qin2022darkforest,daian2020flash,li2023demystifying,qin2021flashloans} has motivated substantial research on order-fairness \cite{li2024sok,ciampi2025uc}. We organize this literature into three layers.

\textbf{Fairness definitions.} Aequitas \cite{kelkar2020order} formalized batch order-fairness through pairwise receive-order constraints, later adopted by Themis. Wendy \cite{kursawe2020wendy} and Pompe \cite{zhang2020pompe} explored timestamp-style fairness with weaker guarantees. Quick Order Fairness \cite{cachin2022quick} studies efficiency boundaries, DCN~\cite{constantinescu2023dcn} uses decentralized clocks, and bounded unfairness \cite{kiayias2024ordering} shows that minimizing order distortion inside cyclic regions is NP-hard. UC Transaction Serialization \cite{ciampi2025uc} elevates fairness into composable cryptographic abstractions. FlashOrder does not redefine fairness; it provides a deterministic execution path for the $\gamma$-batch-order-fairness regime.

\textbf{Ordering engines.} The dominant systems line keeps pairwise evidence but realizes it through graph-based aggregation. Themis \cite{kelkar2023themis} made this approach practical through deferred ordering and batch unspooling, while Rashnu \cite{nagda2024rashnu} reduces graph size by focusing on data-dependent transactions. Phalanx \cite{wang2022phalanx} performs ordered Byzantine consensus through graph-style dependency processing. These systems share the same SCC-based execution path. Subsequent work tries to improve this pipeline: CFTO \cite{nassar2024cfto} reduces communication overhead; Age-Aware Fairness \cite{sokolik2025age} controls tail delay; Dikaios \cite{wang2025dikaios} prunes dependency graphs; EquiBFT \cite{cai2025equibft} uses weighted voting; TEE-based approaches~\cite{stathakopoulou2021tee} rely on stronger trust assumptions. Auncel \cite{chen2024auncel} also targets the SCC bottleneck, but solves it through weight-based linearization ($W = 1 - k^d$) that commits all transactions in a single consensus round, yielding a total order without explicit cycle resolution. FlashOrder takes a different execution-level approach: it restructures ordering so that ambiguity remains localized through adaptive clustering before a batch-wide SCC can form, and uses hierarchical serialization rather than a single global pass. Its intra-cluster Net-Preference Score is related to aggregate preference rules in social choice theory~\cite{brandt2016handbook}, used as a deterministic local refinement rather than an optimal solution to cyclic aggregation.

\textbf{Protocol architectures.} SpeedyFair \cite{mu2024speedyfair} decouples ordering from consensus, a philosophy FlashOrder also adopts. Narwhal/Tusk \cite{danezis2022narwhal}, Bullshark \cite{spiegelman2022bullshark}, Shoal++ \cite{arun2025shoal}, Mysticeti~\cite{babel2023mysticeti}, Sailfish~\cite{shrestha2025sailfish}, Autobahn~\cite{giridharan2024autobahn}, and Jolteon/Ditto~\cite{gelashvili2022jolteon} improve DAG-based BFT throughput but do not solve fair ordering. FairDAG \cite{kang2025fairdag}, DAG of DAGs \cite{nagda2025dag}, and Fides~\cite{xie2025fides} integrate fair ordering into the DAG-based proposer layer, while Travelers~\cite{xue2024travelers} reduces communication complexity for fair ordering through probabilistic techniques. These approaches are complementary: they improve dissemination and proposal, while FlashOrder rewrites the ordering engine that serializes accumulated pairwise evidence into a final order.

\textbf{Layer 2 sequencers.} Most L2 rollups employ centralized sequencers, creating MEV opportunities \cite{alipanahloo2024maximum}. Mitigations include auction-based priority (Timeboost~\cite{mamageishvili2023timeboost}) and decentralized sequencer networks (Espresso~\cite{bearer2024espresso}, Astria). FlashOrder is orthogonal: it targets the algorithmic cost of converting replicated evidence into a deterministic fair order, regardless of who controls sequencing.

\section{Limitations}

FlashOrder uses the same thresholded pairwise evidence model as $\gamma$-batch-order-fair protocols, but its quotient graph aggregates evidence at the cluster level. Therefore, an unconditional statement that every robust pair is preserved across clusters would require an additional separation condition: the pair's signal must either induce the correct quotient edge or be absorbed into the same final cluster. The correctness arguments in Section~\ref{sec:correctness} make this condition explicit, keeping the correctness claim aligned with the algorithmic object FlashOrder actually constructs, while separating it from the empirical quality metrics reported in Section~\ref{sec:eval_fairness_quality}.

In particular, FlashOrder does not claim unconditional $\gamma$-batch-order-fairness for arbitrary cyclic regions. Intra-cluster Net-Preference Score is a deterministic refinement for ambiguous or cyclic regions; it does not claim to satisfy all robust pairs inside arbitrary Condorcet cycles, which is precisely the hard case identified by bounded-unfairness analyses~\cite{kiayias2024ordering}. Its role is to provide a reproducible total order whose fairness quality is measured empirically through robust-pair satisfaction and rank displacement.

\section{Conclusion}
\label{sec:conclusion}
FlashOrder is a deterministic fair-ordering engine that addresses the ordering-engine bottleneck of global SCC condensation in BFT transaction sequencing. By projecting pairwise preferences into a one-dimensional canonical embedding and localizing cyclic ambiguity through adaptive clustering and hierarchical serialization, FlashOrder shifts the critical path from batch-wide graph traversal to localized sorting and aggregation. Across prototype deployment and controlled simulation, FlashOrder achieves up to 10.5$\times$ higher throughput than Themis and 4.8$\times$ higher than Rashnu, with 77.2\% throughput retention relative to the non-fairness baseline at $n{=}101$. Under Condorcet attacks in a 50\,ms WAN, FlashOrder sustains 12.0$\times$ and 9.7$\times$ higher throughput than Themis and Rashnu on average, while reducing ordering latency from 13.0\,s to 0.56\,s over Themis. At $\gamma{=}0.8$, FlashOrder reduces maximum rank displacement from 301 and 287 to 34 positions (an 88.7\% reduction) and achieves an RPS of 0.873, improving to near-perfect RPS of 0.996 and MaxDisp of 10 at the stricter $\gamma{=}0.51$ setting. FlashOrder shows that fair ordering can be made more practical and replay-verifiable by restructuring the ordering engine itself, without changing the pairwise evidence model or redesigning the consensus substrate. Future work will integrate this localized ordering path with DAG-based consensus architectures, connect it more tightly to composable fairness frameworks, and extend it to broader shared-sequencing settings.

\bibliographystyle{ACM-Reference-Format}
\bibliography{ref}

@inproceedings{qin2022darkforest,
  author       = {Kaihua Qin and
                  Liyi Zhou and
                  Arthur Gervais},
  title        = {Quantifying Blockchain Extractable Value: How dark is the forest?},
  booktitle    = {43rd {IEEE} Symposium on Security and Privacy, {SP} 2022, San Francisco,
                  CA, USA, May 22-26, 2022},
  pages        = {198--214},
  publisher    = {{IEEE}},
  year         = {2022},
  doi          = {10.1109/SP46214.2022.9833734}
}

@inproceedings{daian2020flash,
  author       = {Philip Daian and
                  Steven Goldfeder and
                  Tyler Kell and
                  Yunqi Li and
                  Xueyuan Zhao and
                  Iddo Bentov and
                  Lorenz Breidenbach and
                  Ari Juels},
  title        = {Flash Boys 2.0: Frontrunning in Decentralized Exchanges, Miner Extractable
                  Value, and Consensus Instability},
  booktitle    = {2020 {IEEE} Symposium on Security and Privacy, {SP} 2020, San Francisco,
                  CA, USA, May 18-21, 2020},
  pages        = {910--927},
  publisher    = {{IEEE}},
  year         = {2020},
  doi          = {10.1109/SP40000.2020.00040}
}

@inproceedings{zhang2020pompe,
  author       = {Yunhao Zhang and
                  Srinath T. V. Setty and
                  Qi Chen and
                  Lidong Zhou and
                  Lorenzo Alvisi},
  title        = {Byzantine Ordered Consensus without Byzantine Oligarchy},
  booktitle    = {14th {USENIX} Symposium on Operating Systems Design and Implementation,
                  {OSDI} 2020, Virtual Event, November 4-6, 2020},
  pages        = {633--649},
  publisher    = {{USENIX} Association},
  year         = {2020}
}

@inproceedings{kursawe2020wendy,
  author       = {Klaus Kursawe},
  title        = {Wendy, the Good Little Fairness Widget: Achieving Order Fairness for
                  Blockchains},
  booktitle    = {Proceedings of the 2nd {ACM} Conference on Advances in Financial Technologies,
                  {AFT} 2020, New York, NY, USA, October 21-23, 2020},
  pages        = {25--36},
  publisher    = {{ACM}},
  year         = {2020},
  doi          = {10.1145/3419614.3423263}
}

@inproceedings{mu2024speedyfair,
  author       = {Ke Mu and
                  Bo Yin and
                  Alia Asheralieva and
                  Xuetao Wei},
  title        = {Separation is Good: {A} Faster Order-Fairness Byzantine Consensus},
  booktitle    = {31st Annual Network and Distributed System Security Symposium, {NDSS}
                  2024, San Diego, California, USA, February 26 - March 1, 2024},
  publisher    = {The Internet Society},
  year         = {2024}
}

@article{li2024sok,
  author       = {Zhuolun Li and
                  Evangelos Pournaras},
  title        = {SoK: Consensus for Fair Message Ordering},
  journal      = {CoRR},
  volume       = {abs/2411.09981},
  year         = {2024},
  doi          = {10.48550/arXiv.2411.09981},
  eprinttype   = {arXiv},
  eprint       = {2411.09981}
}

@inproceedings{li2023demystifying,
  author       = {Zihao Li and
                  Jianfeng Li and
                  Zheyuan He and
                  Xiapu Luo and
                  Ting Wang and
                  Xiaoze Ni and
                  Wenwu Yang and
                  Xi Chen and
                  Ting Chen},
  title        = {Demystifying {DeFi} {MEV} Activities in Flashbots Bundle},
  booktitle    = {Proceedings of the 2023 {ACM} {SIGSAC} Conference on Computer and
                  Communications Security, {CCS} 2023, Copenhagen, Denmark, November
                  26-30, 2023},
  pages        = {165--179},
  publisher    = {{ACM}},
  year         = {2023},
  doi          = {10.1145/3576915.3616590}
}

@inproceedings{kelkar2023themis,
  author       = {Mahimna Kelkar and
                  Soubhik Deb and
                  Sishan Long and
                  Ari Juels and
                  Sreeram Kannan},
  title        = {Themis: Fast, Strong Order-Fairness in Byzantine Consensus},
  booktitle    = {Proceedings of the 2023 {ACM} {SIGSAC} Conference on Computer and
                  Communications Security, {CCS} 2023, Copenhagen, Denmark, November
                  26-30, 2023},
  pages        = {475--489},
  publisher    = {{ACM}},
  year         = {2023},
  doi          = {10.1145/3576915.3616658}
}

@inproceedings{kelkar2020order,
  author       = {Mahimna Kelkar and
                   Fan Zhang and
                   Steven Goldfeder and
                   Ari Juels},
  title        = {Order-Fairness for Byzantine Consensus},
  booktitle    = {Advances in Cryptology -- {CRYPTO} 2020: 40th Annual International Cryptology
                   Conference, {CRYPTO} 2020, Santa Barbara, CA, USA, August 17-21, 2020,
                   Proceedings, Part {III}},
  series       = {Lecture Notes in Computer Science},
  volume       = {12172},
  pages        = {451--480},
  publisher    = {Springer International Publishing},
  year         = {2020},
  doi          = {10.1007/978-3-030-56877-1_16}
}

@inproceedings{danezis2022narwhal,
  author       = {George Danezis and
                  Lefteris Kokoris{-}Kogias and
                  Alberto Sonnino and
                  Alexander Spiegelman},
  title        = {Narwhal and Tusk: a {DAG}-based mempool and efficient {BFT} consensus},
  booktitle    = {EuroSys '22: Seventeenth European Conference on Computer Systems,
                  Rennes, France, April 5 - 8, 2022},
  pages        = {34--50},
  publisher    = {{ACM}},
  year         = {2022},
  doi          = {10.1145/3492321.3519594}
}

@article{nagda2025dag,
  author       = {Heena Nagda and
                  Sidharth Sankhe and
                  Sakshi Sinha and
                  Keon Attarha and
                  Mohammad Javad Amiri and
                  Boon Thau Loo},
  title        = {{DAG} of {DAGs}: Order-Fairness Made Practical},
  journal      = {Proc. {ACM} Manag. Data},
  volume       = {3},
  number       = {6},
  pages        = {1--27},
  year         = {2025},
  doi          = {10.1145/3769777}
}

@inproceedings{yin2019hotstuff,
  author       = {Maofan Yin and
                  Dahlia Malkhi and
                  Michael K. Reiter and
                  Guy Golan{-}Gueta and
                  Ittai Abraham},
  title        = {HotStuff: {BFT} Consensus with Linearity and Responsiveness},
  booktitle    = {Proceedings of the 2019 {ACM} Symposium on Principles of Distributed
                  Computing, {PODC} 2019, Toronto, ON, Canada, July 29 - August 2, 2019},
  pages        = {347--356},
  publisher    = {{ACM}},
  year         = {2019},
  doi          = {10.1145/3293611.3331591}
}

@inproceedings{cachin2022quick,
  author       = {Christian Cachin and
                  Jovana Micic and
                  Nathalie Steinhauer and
                  Luca Zanolini},
  title        = {Quick Order Fairness},
  booktitle    = {Financial Cryptography and Data Security -- 26th International Conference,
                  {FC} 2022, Grenada, May 2-6, 2022, Revised Selected Papers},
  series       = {Lecture Notes in Computer Science},
  pages        = {316--333},
  publisher    = {Springer},
  year         = {2022},
  doi          = {10.1007/978-3-031-18283-9\_15}
}

@article{alipanahloo2024maximum,
  author       = {Zeinab Alipanahloo and
                  Abdelhakim Senhaji Hafid and
                  Kaiwen Zhang},
  title        = {Maximum Extractable Value {(MEV)} Mitigation Approaches in Ethereum
                  and Layer-2 Chains: {A} Comprehensive Survey},
  journal      = {{IEEE} Access},
  volume       = {12},
  pages        = {185212--185231},
  year         = {2024},
  doi          = {10.1109/ACCESS.2024.3514375}
}

@article{nassar2024cfto,
  author       = {Mohammad Nassar and
                  Ori Rottenstreich and
                  Ariel Orda},
  title        = {{CFTO:} Communication-Aware Fairness in Blockchain Transaction Ordering},
  journal      = {{IEEE} Trans. Netw. Serv. Manag.},
  volume       = {21},
  number       = {1},
  pages        = {490--506},
  year         = {2024},
  doi          = {10.1109/TNSM.2023.3298201}
}

@article{nagda2024rashnu,
  author       = {Heena Nagda and
                  Shubhendra Pal Singhal and
                  Mohammad Javad Amiri and
                  Boon Thau Loo},
  title        = {Rashnu: Data-Dependent Order-Fairness},
  journal      = {Proc. {VLDB} Endow.},
  volume       = {17},
  number       = {9},
  pages        = {2335--2348},
  year         = {2024},
  doi          = {10.14778/3665844.3665861}
}

@inproceedings{cai2025equibft,
  author       = {Siwei Cai and
                  Lei Fan and
                  Shengyun Liu and
                  Hong{-}Sheng Zhou},
  title        = {EquiBFT: {A} Framework for Achieving Fairness in {BFT} Consensus},
  booktitle    = {45th {IEEE} International Conference on Distributed Computing Systems,
                  {ICDCS} 2025, Glasgow, United Kingdom, July 21-23, 2025},
  pages        = {341--351},
  publisher    = {{IEEE}},
  year         = {2025},
  doi          = {10.1109/ICDCS63083.2025.00041}
}

@inproceedings{chen2024auncel,
  author       = {Wuhui Chen and
                  Yikai Feng and
                  Jianting Zhang and
                  Zhongteng Cai and
                  Hong{-}Ning Dai and
                  Zibin Zheng},
  title        = {Auncel: Fair Byzantine Consensus Protocol with High Performance},
  booktitle    = {{IEEE} {INFOCOM} 2024 - {IEEE} Conference on Computer Communications,
                  Vancouver, BC, Canada, May 20-23, 2024},
  pages        = {1251--1260},
  publisher    = {{IEEE}},
  year         = {2024},
  doi          = {10.1109/INFOCOM52122.2024.10621379}
}

@article{wang2025dikaios,
  author       = {Yang Wang and
                  Xiaofei Xing and
                  Guojun Wang and
                  Yuheng Zhang and
                  Peiqiang Li},
  title        = {Dikaios: Position-anchored group ordering with reputation for fair
                  and efficient Byzantine consensus},
  journal      = {Comput. Networks},
  volume       = {269},
  pages        = {111423},
  year         = {2025},
  doi          = {10.1016/j.comnet.2025.111423}
}

@article{sokolik2025age,
  author       = {Yaakov Sokolik and
                   Mohammad Nassar and
                   Ori Rottenstreich},
  title        = {Age-Aware Fairness in Blockchain Transaction Ordering for Reducing
                   Tail Latency},
  journal      = {{IEEE/ACM} Transactions on Networking},
  volume       = {33},
  number       = {2},
  pages        = {807--822},
  year         = {2025},
  doi          = {10.1109/TNET.2024.3503758}
}

@inproceedings{vafadar2023condorcet,
  author       = {Mohammad Amin Vafadar and
                  Majid Khabbazian},
  title        = {Condorcet Attack Against Fair Transaction Ordering},
  booktitle    = {5th Conference on Advances in Financial Technologies, {AFT} 2023,
                  Princeton, NJ, USA, October 23-25, 2023},
  series       = {LIPIcs},
  pages        = {15:1--15:21},
  publisher    = {Schloss Dagstuhl - Leibniz-Zentrum f{\"{u}}r Informatik},
  year         = {2023},
  doi          = {10.4230/LIPIcs.AFT.2023.15}
}

@inproceedings{park2025frontrunning,
  author       = {Eunchan Park and
                  Taeung Yoon and
                  Hocheol Nam and
                  Deepak Maram and
                  Min Suk Kang},
  title        = {On Frontrunning Risks in Batch-Order Fair Systems for Blockchains},
  booktitle    = {Proceedings of the 2025 {ACM} {SIGSAC} Conference on Computer and
                  Communications Security, {CCS} 2025, Taipei, Taiwan, October 13-17,
                  2025},
  pages        = {918--932},
  publisher    = {{ACM}},
  year         = {2025},
  doi          = {10.1145/3719027.3744879}
}

@inproceedings{cloudlab,
  author       = {Dmitry Duplyakin and
                  Robert Ricci and
                  Aleksander Maricq and
                  Gary Wong and
                  Jonathon Duerig and
                  Eric Eide and
                  Leigh Stoller and
                  Mike Hibler and
                  David Johnson and
                  Kirk Webb and
                  Aditya Akella and
                  Kuang{-}Ching Wang and
                  Glenn Ricart and
                  Larry Landweber and
                  Chip Elliott and
                  Michael Zink and
                  Emmanuel Cecchet and
                  Snigdhaswin Kar and
                  Prabodh Mishra},
  title        = {The Design and Operation of CloudLab},
  booktitle    = {Proceedings of the 2019 {USENIX} Annual Technical Conference, {USENIX}
                  {ATC} 2019, Renton, WA, USA, July 10-12, 2019},
  pages        = {1--14},
  publisher    = {{USENIX} Association},
  year         = {2019}
}

@misc{bearer2024espresso,
  author       = {Jeb Bearer and
                  Benedikt B{\"u}nz and
                  Philippe Camacho and
                  Binyi Chen and
                  Ellie Davidson and
                  Ben Fisch and
                  Brendon Fish and
                  Gus Gutoski and
                  Fernando Krell and
                  Chengyu Lin and
                  Dahlia Malkhi and
                  Kartik Nayak and
                  Keyao Shen and
                  Alex Xiong and
                  Nathan Yospe and
                  Sishan Long},
  title        = {The Espresso Sequencing Network: {HotShot} Consensus, Tiramisu
                  Data-Availability, and Builder-Exchange},
  howpublished = {Cryptology {ePrint} Archive, Paper 2024/1189},
  year         = {2024},
  url          = {https://eprint.iacr.org/2024/1189}
}

@inproceedings{mamageishvili2023timeboost,
  author       = {Akaki Mamageishvili and
                  Mahimna Kelkar and
                  Jan Christoph Schlegel and
                  Edward W. Felten},
  title        = {Buying Time: Latency Racing vs. Bidding for Transaction Ordering},
  booktitle    = {5th Conference on Advances in Financial Technologies, {AFT} 2023,
                  Princeton, NJ, USA, October 23-25, 2023},
  series       = {Leibniz International Proceedings in Informatics (LIPIcs)},
  volume       = {282},
  pages        = {23:1--23:22},
  publisher    = {Schloss Dagstuhl -- Leibniz-Zentrum f{\"{u}}r Informatik},
  year         = {2023},
  doi          = {10.4230/LIPIcs.AFT.2023.23}
}

@inproceedings{castro1999practical,
  author       = {Miguel Castro and
                  Barbara Liskov},
  title        = {Practical Byzantine Fault Tolerance},
  booktitle    = {Proceedings of the 3rd Symposium on Operating Systems Design and
                  Implementation, {OSDI} 1999, New Orleans, Louisiana, USA, February
                  22-25, 1999},
  pages        = {173--186},
  publisher    = {{USENIX} Association},
  year         = {1999}
}

@article{wang2022phalanx,
  author       = {Guangren Wang and
                   Liang Cai and
                   Fangyu Gai and
                   Jianyu Niu},
  title        = {Phalanx: A Practical Byzantine Ordered Consensus Protocol},
  journal      = {CoRR},
  volume       = {abs/2209.08512},
  year         = {2022},
  eprinttype   = {arXiv},
  eprint       = {2209.08512}
}

@inproceedings{kiayias2024ordering,
  title={Ordering transactions with bounded unfairness: Definitions, complexity and constructions},
  author={Kiayias, Aggelos and Leonardos, Nikos and Shen, Yu},
  booktitle={Annual International Conference on the Theory and Applications of Cryptographic Techniques},
  pages={34--63},
  year={2024},
  organization={Springer}
}

@article{ren2025proofcarrying,
  author       = {Pengkun Ren and
                  Hai Dong and
                  Nasrin Sohrabi and
                  Zahir Tari and
                  Pengcheng Zhang},
  title        = {Proof-Carrying Fair Ordering: Asymmetric Verification for {BFT}
                  via Incremental Graphs},
  journal      = {CoRR},
  volume       = {abs/2510.14186},
  year         = {2025},
  eprinttype   = {arXiv},
  eprint       = {2510.14186}
}

@inproceedings{spiegelman2022bullshark,
  author       = {Alexander Spiegelman and
                  Neil Giridharan and
                  Alberto Sonnino and
                  Lefteris Kokoris{-}Kogias},
  title        = {Bullshark: {DAG} {BFT} Protocols Made Practical},
  booktitle    = {Proceedings of the 2022 {ACM} {SIGSAC} Conference on Computer and
                  Communications Security, {CCS} 2022, Los Angeles, CA, USA, November
                  7-11, 2022},
  pages        = {2705--2718},
  publisher    = {{ACM}},
  year         = {2022},
  doi          = {10.1145/3548606.3559361}
}

@inproceedings{arun2025shoal,
  author       = {Balaji Arun and
                   Zekun Li and
                   Florian Suri{-}Payer and
                   Sourav Das and
                   Alexander Spiegelman},
  title        = {Shoal++: High Throughput {DAG} {BFT} Can Be Fast and Robust!},
  booktitle    = {22nd {USENIX} Symposium on Networked Systems Design and Implementation,
                   {NSDI} 2025, Philadelphia, PA, USA, April 28-30, 2025},
  publisher    = {{USENIX} Association},
  year         = {2025},
  pages        = {813--826}
}

@inproceedings{miller2016honeybadger,
  author       = {Andrew Miller and
                   Yu Xia and
                   Kyle Croman and
                   Elaine Shi and
                   Dawn Song},
  title        = {The {HoneyBadger} of {BFT} Protocols},
  booktitle    = {Proceedings of the 2016 {ACM} {SIGSAC} Conference on Computer and
                   Communications Security, {CCS} 2016, Vienna, Austria, October 24-28,
                   2016},
  pages        = {31--42},
  publisher    = {{ACM}},
  year         = {2016},
  doi          = {10.1145/2976749.2978399}
}

@inproceedings{gilad2017algorand,
  author       = {Yossi Gilad and
                  Rotem Hemo and
                  Silvio Micali and
                  Georgios Vlachos and
                  Nickolai Zeldovich},
  title        = {Algorand: Scaling {Byzantine} Agreements for Cryptocurrencies},
  booktitle    = {Proceedings of the 26th Symposium on Operating Systems Principles,
                  {SOSP} 2017, Shanghai, China, October 28-31, 2017},
  pages        = {51--68},
  publisher    = {{ACM}},
  year         = {2017},
  doi          = {10.1145/3132747.3132757}
}

@article{team2018avalanche,
  author       = {Team Rocket and
                   Maofan Yin and
                   Kevin Sekniqi and
                   Robbert van Renesse and
                   Emin G{\"u}n Sirer},
  title        = {Scalable and Probabilistic Leaderless {BFT} Consensus through Metastability},
  journal      = {CoRR},
  volume       = {abs/1906.08936},
  year         = {2020},
  eprinttype   = {arXiv},
  eprint       = {1906.08936}
}

@article{kang2025fairdag,
  author       = {Dakai Kang and
                   Junchao Chen and
                   Tien Tuan Anh Dinh and
                   Mohammad Sadoghi},
  title        = {{FairDAG}: Consensus Fairness over Multi-Proposer Causal Design},
  journal      = {Proc. {VLDB} Endow.},
  volume       = {19},
  number       = {2},
  pages        = {265--278},
  year         = {2025},
  doi          = {10.14778/3773749.3773763}
}

@article{thomson2010determinism,
  author       = {Alexander Thomson and
                  Daniel J. Abadi},
  title        = {The Case for Determinism in Database Systems},
  journal      = {Proc. {VLDB} Endow.},
  volume       = {3},
  number       = {1-2},
  pages        = {70--80},
  year         = {2010},
  doi          = {10.14778/1920841.1920855}
}

@inproceedings{ruan2021blockchain,
  author       = {Pingcheng Ruan and
                  Tien Tuan Anh Dinh and
                  Dumitrel Loghin and
                  Meihui Zhang and
                  Gang Chen and
                  Qian Lin and
                  Beng Chin Ooi},
  title        = {Blockchains vs. Distributed Databases: Dichotomy and Fusion},
  booktitle    = {Proc. {ACM} {SIGMOD}/\{PODS\} Conf.},
  pages        = {1504--1517},
  year         = {2021},
  doi          = {10.1145/3448016.3452789}
}

@book{condorcet1785,
  author       = {Marquis de Condorcet},
  title        = {Essai sur l'application de l'analyse {\`a} la probabilit{\'e} des d{\'e}cisions rendues {\`a} la pluralit{\'e} des voix},
  publisher    = {Imprimerie Royale},
  address      = {Paris},
  year         = {1785}
}

@article{buchman2018tendermint,
  author       = {Ethan Buchman and
                  Jae Kwon and
                  Zarko Milosevic},
  title        = {The Latest Gossip on {BFT} Consensus},
  journal      = {CoRR},
  volume       = {abs/1807.04938},
  year         = {2018},
  eprinttype   = {arXiv},
  eprint       = {1807.04938}
}

@inproceedings{keidar2021dagrider,
  author       = {Idit Keidar and
                  Eleftherios Kokoris{-}Kogias and
                  Oded Naor and
                  Alexander Spiegelman},
  title        = {All You Need is {DAG}},
  booktitle    = {Proceedings of the 2021 {ACM} Symposium on Principles of Distributed Computing,
                  {PODC} 2021, Virtual Event, Italy, July 26-30, 2021},
  pages        = {165--175},
  publisher    = {{ACM}},
  year         = {2021},
  doi          = {10.1145/3465084.3467905}
}

@inproceedings{constantinescu2023dcn,
  author       = {Andrei Constantinescu and
                   Diana Ghinea and
                   Lioba Heimbach and
                   Zilin Wang and
                   Roger Wattenhofer},
  title        = {A Fair and Resilient Decentralized Clock Network for Transaction Ordering},
  booktitle    = {27th International Conference on Principles of Distributed Systems, {OPODIS}
                   2023, Tokyo, Japan, December 6-8, 2023},
  series       = {Leibniz International Proceedings in Informatics (LIPIcs)},
  volume       = {286},
  pages        = {8:1--8:20},
  publisher    = {Schloss Dagstuhl -- Leibniz-Zentrum f{\"{u}}r Informatik},
  year         = {2024},
  doi          = {10.4230/LIPIcs.OPODIS.2023.8}
}

@inproceedings{babel2023mysticeti,
  author       = {Kushal Babel and
                   Andrey Chursin and
                   George Danezis and
                   Anastasios Kichidis and
                   Lefteris Kokoris{-}Kogias and
                   Arun Koshy and
                   Alberto Sonnino and
                   Mingwei Tian},
  title        = {Mysticeti: Reaching the Latency Limits with Uncertified {DAGs}},
  booktitle    = {Proceedings of the 2025 Network and Distributed System Security Symposium,
                   {NDSS} 2025, San Diego, CA, USA, February 24-28, 2025},
  publisher    = {The Internet Society},
  year         = {2025},
  doi          = {10.14722/ndss.2025.240929}
}

@inproceedings{shrestha2025sailfish,
  author       = {Nibesh Shrestha and
                  Ruchit Shrothrium and
                  Aniket Kate and
                  Kartik Nayak},
  title        = {Sailfish: Towards Improving the Latency of {DAG}-based {BFT}},
  booktitle    = {46th {IEEE} Symposium on Security and Privacy, {SP} 2025},
  publisher    = {{IEEE}},
  year         = {2025},
  doi          = {10.1109/SP61157.2025.00021}
}

@inproceedings{giridharan2024autobahn,
  author       = {Neil Giridharan and
                  Florian Suri{-}Payer and
                  Ittai Abraham and
                  Lorenzo Alvisi and
                  Natacha Crooks},
  title        = {Autobahn: Seamless High Speed {BFT}},
  booktitle    = {Proceedings of the {ACM} {SIGOPS} 30th Symposium on Operating Systems Principles,
                  {SOSP} 2024},
  pages        = {1--23},
  publisher    = {{ACM}},
  year         = {2024}
}

@inproceedings{ciampi2025uc,
  author       = {Michele Ciampi and
                   Aggelos Kiayias and
                   Yu Shen},
  title        = {Universal Composable Transaction Serialization with Order Fairness},
  booktitle    = {Advances in Cryptology -- {CRYPTO} 2024: 44th Annual International Cryptology
                   Conference, Santa Barbara, CA, USA, August 18-22, 2024, Proceedings, Part {I}},
  series       = {Lecture Notes in Computer Science},
  volume       = {14921},
  pages        = {147--180},
  publisher    = {Springer Nature Switzerland},
  year         = {2024},
  doi          = {10.1007/978-3-031-68379-4\_5}
}

@book{brandt2016handbook,
  author       = {Felix Brandt and
                  Vincent Conitzer and
                  Ulle Endriss and
                  J{\'e}r{\^o}me Lang and
                  Ariel D. Procaccia},
  title        = {Handbook of Computational Social Choice},
  publisher    = {Cambridge University Press},
  year         = {2016}
}

@inproceedings{stathakopoulou2021tee,
  author       = {Chrysoula Stathakopoulou and
                  Signe R{\"u}sch and
                  Marcus Brandenburger and
                  Marko Vukoli{\'c}},
  title        = {Adding Fairness to Order: Preventing Front-Running Attacks in {BFT} Protocols
                  Using {TEE}s},
  booktitle    = {40th {IEEE} International Symposium on Reliable Distributed Systems, {SRDS} 2021},
  pages        = {34--45},
  publisher    = {{IEEE}},
  year         = {2021},
  doi          = {10.1109/SRDS53918.2021.00013}
}

@article{xie2025fides,
  author       = {Shaokang Xie and
                  Dakai Kang and
                  Hanzheng Lyu and
                  Jianyu Niu and
                  Mohammad Sadoghi},
  title        = {Fides: Scalable Censorship-Resistant {DAG} Consensus via Trusted Components},
  journal      = {CoRR},
  volume       = {abs/2501.01062},
  year         = {2025},
  eprinttype   = {arXiv},
  eprint       = {2501.01062}
}

@inproceedings{gelashvili2022jolteon,
  author       = {Rati Gelashvili and
                  Lefteris Kokoris{-}Kogias and
                  Alberto Sonnino and
                  Alexander Spiegelman and
                  Zhuolun Xiang},
  title        = {Jolteon and Ditto: Network-Adaptive Efficient Consensus with Asynchronous
                  Fallback},
  booktitle    = {Financial Cryptography and Data Security -- 26th International Conference,
                  {FC} 2022, Grenada, May 2-6, 2022},
  series       = {Lecture Notes in Computer Science},
  pages        = {296--315},
  publisher    = {Springer},
  year         = {2022},
  doi          = {10.1007/978-3-031-18283-9\_14}
}

@inproceedings{qin2021flashloans,
  author       = {Kaihua Qin and
                  Liyi Zhou and
                  Benjamin Livshits and
                  Arthur Gervais},
  title        = {Attacking the {DeFi} Ecosystem with Flash Loans for Fun and Profit},
  booktitle    = {Financial Cryptography and Data Security -- 25th International Conference,
                  {FC} 2021, Virtual Event, March 1-5, 2021},
  series       = {Lecture Notes in Computer Science},
  pages        = {3--32},
  publisher    = {Springer},
  year         = {2021}
}

@article{xue2024travelers,
  title={Travelers: A scalable fair ordering BFT system},
  author={Xue, Bowen and Kannan, Sreeram},
  journal={arXiv preprint arXiv:2401.02030},
  year={2024}
}

\end{document}